\documentclass[authoryear, 12pt]{elsarticle}
\usepackage[T1]{fontenc}
\usepackage[latin9]{inputenc}
\usepackage{float}
\usepackage{bm}
\usepackage{amsthm}
\usepackage{amsmath}
\usepackage{amssymb}
\usepackage{graphicx}
\usepackage{esint}
\usepackage{cases}
\usepackage[english]{babel}
\usepackage{subcaption}

\makeatletter

\floatstyle{ruled}
\newfloat{algorithm}{tbp}{loa}
\providecommand{\algorithmname}{Algorithm}
\floatname{algorithm}{\protect\algorithmname}

\theoremstyle{plain}
\newtheorem{thm}{\protect\theoremname}
  \theoremstyle{plain}
  \newtheorem{lem}[thm]{\protect\lemmaname}
  \theoremstyle{plain}
\newtheorem{cor}{\protect\corollaryname}

\usepackage{dsfont} 

\usepackage{algpseudocode}
\usepackage{mathtools}

\newcommand{\ui}{[0,1)} 

\newcommand{\E}{\mathbb{E}}
\newcommand{\var}{\mathrm{Var}}

\newcommand{\bx}{\mathbf{x}}

\newcommand{\bz}{\mathbf{z}}

\newcommand{\dd}{\mathrm{d}}
\newcommand{\dx}{\dd \mathbf{x}}

\newcommand{\bigO}{\mathcal{O}} 
\newcommand{\smallo}{{\scriptscriptstyle\mathcal{O}}} 
\newcommand{\cvz}{\rightarrow 0} 
\renewcommand{\emptyset}{\varnothing} 

\newcommand{\comment}[1]{ \ifthenelse{ \equal{\showcomment}{true} }{ {\bf #1} }{} }
\newcommand{\showcomment}{true}

\makeatother

\usepackage{babel}
  \providecommand{\lemmaname}{Lemma}
\providecommand{\theoremname}{Theorem}
\providecommand{\corollaryname}{Corollary}
\newtheorem{prop}{Proposition}

\usepackage[hmargin=1 in]{geometry}

\let\originaleqref\eqref
\renewcommand{\eqref}{Eq.~\originaleqref}

\begin{document}

\begin{frontmatter}

\title{On Integration Methods Based on Scrambled Nets of Arbitrary Size}

\author{Mathieu Gerber\footnote{Present address: Department of Statistics, Harvard University, Science Center 7th floor, One Oxford Street, Cambridge, MA 02138 (mathieugerber@fas.harvard.edu)}}
\address{Faculty of Business and Economics, Internef building, Universit\'e de Lausanne, 1\,015 Lausanne, Switzerland}
\address{Lab of Statistics, CREST,  
3 avenue Pierre Larousse, 92\,240 Malakoff, France}

\date{}

\begin{abstract}
We consider the problem of evaluating $I(\varphi):=\int_{\ui^s}\varphi(\bx) \dx$ for a function $\varphi \in L^2[0,1)^{s}$. In situations where $I(\varphi)$ can be approximated by an estimate of the form $N^{-1}\sum_{n=0}^{N-1}\varphi(\bx^n)$, with $\{\bx^n\}_{n=0}^{N-1}$ a point set in $\ui^s$,  
it is now well known that the $\bigO_P(N^{-1/2})$ Monte Carlo convergence rate  can be improved by taking for $\{\bx^n\}_{n=0}^{N-1}$ the first $N=\lambda b^m$ points, $\lambda\in\{1,\dots,b-1\}$, of a scrambled $(t,s)$-sequence in base $b\geq 2$. In this paper we  derive a bound for the variance of scrambled net quadrature rules  which is of order $\smallo(N^{-1})$ without any restriction on  $N$. As a corollary, this bound  allows us to   provide simple conditions to get, for any pattern of $N$, an integration error of size $\smallo_P(N^{-1/2})$ for functions that  depend on the quadrature size $N$. Notably, we establish  that  sequential quasi-Monte Carlo (M. Gerber and N. Chopin, 2015, \emph{J. R. Statist. Soc. B}, 77 (3), 509-579) reaches the $\smallo_P(N^{-1/2})$  convergence rate for any values of $N$. In a numerical study, we show that for scrambled net quadrature rules  we can relax the constraint  on $N$ without any loss of efficiency when the integrand $\varphi$ is a discontinuous function while, for  sequential quasi-Monte Carlo, taking $N=\lambda b^m$ may only provide moderate gains.
\end{abstract}


\end{frontmatter}

\section{Introduction}

We consider the problem of evaluating $I(\varphi):=\int_{\ui^s}\varphi(\bx)\dx$
for a function $\varphi \in L^2[0,1)^{s}$. Focussing first on unweighed quadrature rules of the form  $I(P^N,\varphi)=N^{-1}\sum_{n=0}^{N-1}\varphi(\bx^n)$,
with $P^N=\{\bx^n\}_{n=0}^{N-1}$  a  set of $N$ points in $\ui^s$, the simplest way to approximate $I(\varphi)$ is to use the Monte Carlo estimator which selects for $P^N$ a set of $N$ independent uniform random variates on $\ui^s$. The central limit theorem  then ensures that the variance of the approximation error  $I(P^N,\varphi)-I(\varphi)$ is of order $\bigO(N^{-1})$. However, it is now well known that this rate can be improved by taking  for $P^N$  a randomized quasi-Monte Carlo (RQMC) point set. In particular, \citet{Owen1995} proposes a randomization scheme for $(t,s)$-sequences in base $b\geq 2$, known as nested scrambling, such that the variance of the  quadrature rule $I(P^N,\varphi)$ decreases faster than  $N^{-1}$ when $P^N$ is the set made of the first $N$ points of the resulting randomized sequence \citep{Owen1997a,Owen1998}.  \citet{Owen1997a,Owen1998} also establishes that, in this case, $\var(I(P^N,\varphi))\leq c_tN^{-1}\sigma^2$ for a constant $c_t<\infty$ independent of $\varphi$ and where $N^{-1}\sigma^2=N^{-1}\int_{\ui^s} (\varphi(\bx)-I(\varphi))^2\dx$ is the variance of a Monte Carlo quadrature rule of the same size. Interestingly, \citet{Owen1997a} shows that the constant $c_0$ has the additional property to be independent of the dimension $s$.

In some complicated settings, the function $\varphi$ cannot be computed explicitly and/or the dimension $s$ is too large for a simple unweighted quadrature rule $I(P^N,\varphi)$ to be efficient. Important examples where such a problem arises  are parameter and state inference in state space models. Recently, \citet{SQMC} have developed a sequential quasi-Monte Carlo (SQMC) algorithm to carry out sequential  inference in this class of models.  When this algorithm uses points taken from scrambled $(t,s)$-sequences as inputs, it outperforms Monte Carlo methods with an error  of size $\smallo_P(N^{-1/2})$ for continuous and bounded functions  \citep[][Theorem 7]{SQMC}.

However, all these results  apply only for $N=\lambda b^m$, $\lambda\in\{1,\dots,b-1\}$. This restriction on the values of $N$ arises because the approximation error of the aforementioned integration methods  depends on the equidistribution properties of the scrambled nets at hand and, as we go through a scrambled $(t,s)$-sequences in base $b$, sets with the strongest equidistribution properties are constituted of $b^m$ consecutive points, $m\geq t$ (see Section \ref{sec:background} for a review on $(t,s)$-sequences). From a practical point of view, this means that  a (large) variance reduction  can only be obtained at the price of a sharply increasing running time, which may reduce the attractiveness of scrambled net integration methods when one is interested, e.g., to reach a given level of precision at the lowest computational effort.

The objective of this paper is to study quadrature rules and SQMC based on scrambled nets of arbitrary size. Our main theoretical contribution is to provide a bound for the variance of the scrambled net quadrature rule $I(P^N,\varphi)$ which shows  that the $\smallo(N^{-1})$ convergence rate    obtained by \citet{Owen1997a,Owen1998} 
under the restriction $N=\lambda b^m$ in fact holds for any pattern $N$. This bound also provides conditions to have an error of size $\smallo_P(N^{-1/2})$ for the integral of a function   $\varphi_N$ which  depends on the quadrature size $N$, as it typically happens  in sequential estimation methods. A consequence of this last result is  the asymptotic superiority of  SQMC over  sequential Monte Carlo  algorithms without any restriction on $N$. Relaxing the constraint  $N=b^m$ is particularly important for SQMC because in many applications (such as, e.g., target tracking)  inference in state space models should be carried in real time and, consequently, it may be too costly to  double the number of simulations in order to reduce the variance (assuming $b=2$). Having a free control of $N$  is also crucial for parameter inference in state space models if, e.g., one wants to use SQMC as a sampling strategy inside particle Markov Chain Monte Carlo   methods \citep{PMCMC}. Indeed, efficient allocations of the computational budget between the time spent to run the filtering algorithm and the length of the Markov chain require a fine control of $N$ as explained, e.g., in \citet{Doucet2013}.

In addition to a variance of order  $\smallo(N^{-1})$, we show two interesting properties of scrambled net quadrature rules of arbitrary size. First, when points of a scrambled $(0,s)$-sequence are used, the variance of the quadrature rule admits a bound of the form $c^*_0\sigma^2N^{-1}$ for an explicit constant $c^*_0>0$ which is independent of  the integrand $\varphi$ and of the dimension $s$. Second, \citet[][Theorem 4]{Yue1999a} establish that for smooth integrands the integration error of quadratures based on scrambled sequences is of order  $\bigO_P(N^{-1}(\log N)^{(s-1)/2})$. We note in this work that for such functions the error  is in fact of size  $\bigO_P(N^{-1})$. In a recent paper, \citet{Owen2014}  has shown that this rate is the best  we can achieve uniformly in  $N$ for equally weighted  quadrature rules and therefore, on this class of functions, quadratures based on scrambled sequences have the optimal worst case behaviour.

The rest of this paper is organized as follows. Section \ref{sec:background} gives the notation and the  background material used in this work. The announced results for  quadrature rules $I(P^N,\varphi)$ based on scrambled nets are formally stated in Section \ref{sec:res}. In Section \ref{sec:size} we provide  conditions to get the $\smallo_P(N^{-1/2})$ convergence rate for integrands that depend on $N$ and discuss the application of this result in the context of  SQMC. To simplify the presentation, we propose in this section a convergence result for a scrambled net version of the sampling importance resampling (SIR) algorithm introduced by \citet{Rubin1987,Rubin:SIR} rather than for SQMC. This SIR algorithm based on scrambled nets is sequentially used in SQMC and the steps to prove its error rate are exactly the same as the ones needed to relax the constrain on $N$ in  \citet[][Theorem 7]{SQMC}. In Section \ref{sec:num} the question of the impact of $N$ on the convergence rate  for both scrambled nets quadrature rules and  for SQMC  is analysed in a numerical study while Section \ref{sec:conc} concludes.

\section{Background}\label{sec:background}

In this section we provide the background material on $(t,s)$-sequences, scrambled sequences and on the Haar-like decomposition of $L^2\ui^s$ introduced by \citet{Owen1997a}. Only the concepts and the results used in this paper are presented. For a complete exposition of these notions we refer the reader, respectively, to \citet[][Chapter 4]{dick2010digital}, \citet{Owen1995} and \citet{Owen1997a,Owen1998}. 

For integers $s\geq 1$ and $b\geq 2$, let
$$
\mathcal{E}^b=\left\{\prod_{j=1}^s\left[a_j b^{-d_j},(a_j+1)b^{-d_j}\right)\subseteq \ui^s,\, a_j,\,d_j\in\mathbb{N},\, a_j< b^{d_j},
\, j=1,...,s\right\} 
$$
be the set of all $b$-ary boxes.

Let $t$ and $m$ be two positive  integers such that $m\geq t$. Then, the point set $\{\bx^n\}_{n=0}^{b^m-1}$ is called a $(t,m,s)$-net in base $b$ if every $b$-ary box of volume $b^{t-m}$ contains exactly $b^t$ points, while the point set $\{\bx^n\}_{n=0}^{\lambda b^{m}-1}$, $\lambda\in\{1,\dots,b-1\}$, is called a $(\lambda,t,m,s)$-net if every $b$-ary box of volume $b^{t-m}$ contains exactly $\lambda b^t$ points and no $b$-ary box  of volume $b^{t-m-1}$ contains more than $b^t$ points. A sequence  $(\bx^n)_{n\geq 0}$ of points in $\ui^s$ is  called a $(t,s)$-sequence in base $b\geq 2$ if, for any integers $a\geq 0$ and $m\geq  t$, the point set $\{\bx^n\}_{n=ab^m}^{(a+1)b^m-1}$ is a $(t,m,s)$-net in base $b$.  Finally, note that if $(\bx^n)_{n\geq 0}$ is a $(t,s)$-sequence in base $b$, then, for $\lambda\in\{1,\dots,b-1\}$, $\{\bx^n\}_{n=ab^{m+1}}^{ab^{m+1}+\lambda b^m-1}$ is a  $(\lambda,t,m,s)$-net for any integers $a\geq 0$ and $m\geq t$.

To introduce the Haar-like decomposition of $L^2\ui^s$ developed by \citet{Owen1997a}, let $u\subseteq\mathcal{S}:=\{1,...,s\}$, $\kappa$ be a vector of  $|u|$ non negative integers $k_{(u,j)}$, $j\in \{1,\dots,|u|\}$, $|\kappa|=\sum_{j=1}^{|u|}k_{(u,j)}$, and 
$$
\mathcal{E}^b_{u,\kappa}=\left\{\prod_{j=1}^s\left[a_j b^{-d_j},(a_j+1)b^{-d_j}\right)\in\mathcal{E}^b:\,d_j=k_{(u,j)}+1\text{ if } j\in u \text{ and } d_j=0 \text{ if }  j\notin u\right\}.
$$
Then, \citet{Owen1997a} shows that $\varphi(\bx)=\sum_{u\subseteq\mathcal{S}}\sum_{\kappa}
\nu_{u,\kappa}(\bx)$
where, for any $u\subseteq\mathcal{S}$, we use the shorthand $\sum_{\kappa}=\sum_{k_{(u,1)}=0}^{\infty}\dots\sum_{k_{(u,|u|)}=0}^{\infty}$ and $\nu_{u,\kappa}$ is a step function, constant over each of the $b^{|u|+|\kappa|}$ sets $E\in \mathcal{E}^b_{u,\kappa}$ and which integrates  to zero over any $b$-ary box that strictly contains a set $E\in \mathcal{E}^b_{u,\kappa}$. These step functions are mutually orthogonal and $\nu_{\emptyset,()}$ is constant over $\ui^s$. The resulting ANOVA decomposition of $\varphi$ is given by
\begin{equation}\label{eq:sigma}
\sigma^2=\sum_{|u|>0}\sum_{\kappa}\sigma_{u,\kappa}^2
\end{equation}
with $\sigma_{u,\kappa}^2=\int_{\ui^s} \nu_{u,\kappa}^2(\bx)\dx$.

Let $P^N=\{\bx^n\}_{n=0}^{N-1}$, $\bx^n=(x_1^n,\dots,x_s^n)$, be the first $N\geq 1$ points of a $(t,s)$-sequence in base $b\geq 2$ where, for $j\in\mathcal{S}$, $x_j^n=\sum_{i=1}^{\infty}a_{jni}b^{-i}$ with $a_{jni}\in\{0,\dots,b-1\}$ for all $n$ and $i$.  \citet{Owen1995} proposes a method to randomly  permute the digits $a_{jnk}$   such that the scrambled point set $\tilde{P}^N=\{\tilde{\bx}^n\}_{n=0}^{N-1}$  preserves almost surely the equidistribuion properties of the original net $P^N$. In addition, under this randomization scheme, each $\tilde{\bx}^n$ is marginally uniformly distributed on $\ui^s$ and \citet{Owen1997a} shows that
\begin{equation}\label{eq:gen}
\var\left(I(\tilde{P}^N,\varphi)\right)=\frac{1}{N}\sum_{|u|>0}\sum_{\kappa}\Gamma_{u,\kappa}\sigma^2_{u,\kappa}
\end{equation}
where $\Gamma_{u,\kappa}$ depends on the properties of the non scrambled point set $\{\bx^{n}\}_{n=0}^{N-1}$. In particular, for an arbitrary value of $N\in\mathbb{N}^*$, the gain factors $\Gamma_{u,\kappa}$  are bounded by \citep[][Lemma 11]{Hickernell2001} 
\begin{equation}\label{eq:genG}
\Gamma_{u,\kappa}\leq b^{t+1}\left(\frac{b+1}{b-1}\right)^{s+1}.
\end{equation}
When the point set  $P^N$ is a  $(\lambda,t,m,s)$-net, the gain factors  can be more precisely controlled. Notably, \citet[][Lemma 2]{Owen1998} obtains
\begin{align}\label{res}
\var\left(I(\tilde{P}^{N},\varphi)\right)=\frac{1}{N}\sum_{|u|>0}\sum_{|\kappa|>m-t-|u|} \Gamma_{u,\kappa}\sigma^2_{u,\kappa}
\end{align}
where    $\Gamma_{u,\kappa}\leq \Gamma^{(b)}_{t,s}$ with $\Gamma^{(b)}_{0,s}=e$ if $b\geq s$ (\citealp[][Theorem 1]{Owen1997b}; \citealp[][Lemma 6]{Hickernell2001}) and, for $t>0$, $\Gamma^{(b)}_{t,s}=b^t(b+1)^s/(b-1)^s$ \citep[][Lemma 4]{Owen1998}. Together with equation \eqref{res}, these bounds  for the gain factors  imply that
\begin{align}\label{target}
\var\left(I(\tilde{P}^{N},\varphi)\right)=\smallo(N^{-1}),\quad  \var\left(I(\tilde{P}^{N},\varphi)\right)\leq \Gamma^{(b)}_{t,s}\frac{\sigma^2}{N}
\end{align}
where we recall that $\tilde{P}^{N}$ contains the first $N=\lambda b^m$ points of a scrambled $(t,s)$-sequence in base $b\geq 2$. 

We conclude this section by noting that all the results presented in this work also hold for the computationally cheaper scrambling method proposed by \citet{Matousek1998}, although  in what follows we will only refer  to the scrambling technique developed by \citet{Owen1995}  for ease of presentation. In addition, even if it is not always explicitly mentioned, all the scrambled nets we consider in this paper are  made of the first $N$ points of a scrambled $(t,s)$-sequence.

\section{Quadratures based on scrambled nets of arbitrary size}\label{sec:res}

\subsection{Error bounds}

A first result concerning the error bound of quadratures based on scrambled nets of an arbitrary size $N\geq 1$  can be directly deduced from \eqref{eq:gen} and \eqref{eq:genG}. Indeed, if $\tilde{P}^{N}$  contains the first  $N\in\mathbb{N}^*$ points of a scrambled $(t,s)$-sequence in base $b\geq 2$, these two bounds imply  that
\begin{equation}\label{eq:basic}
\var\left(I(\tilde{P}^{N},\varphi)\right)\leq \frac{\sigma^2}{N}b^{t+1}\left(\frac{b+1}{b-1}\right)^{s+1}=\Gamma^{(b)}_{t+1,s+1}\frac{\sigma^2}{N}
\end{equation}
so that the variance of a scrambled net quadrature is never larger than a constant times the Monte Carlo variance. However, this bound is larger than the one in \eqref{target} obtained under the restriction $N=\lambda b^m$  because the equidistribution properties of $\tilde{P}^N$ are the  strongest when $N$ satisfies this constraints. 

The following theorem is the main result of this work and  provides a sharper bound (for $N$ large enough) for the integration error (see \ref{app:Thm} for a proof).

\begin{thm}\label{thm:GenOwen} Let $\varphi \in L^2[0,1)^{s}$,
$\sigma^{2}=\int_{\ui^{s}}\varphi^{2}(\bx)\dx-\big(\int_{\ui^{s}}\varphi(\bx)\dx\big)^{2}$
and $\tilde{P}^N=\{\tilde{\bx}^{n}\}_{n=0}^{N-1}$
be the first $N\in\mathbb{N}^*$ points of a $(t,s)$-sequence in base $b\geq 2$ scrambled as in \citet{Owen1995}. Let $N\geq 1$ and $k\in\mathbb{N}$ be such that $b^k\leq N<b^{k+1}$. Then,
\begin{align*}
\var\big(I(\tilde{P}^{N},\varphi)\big)&\leq 2\frac{\Gamma_{t,s}^{(b)}}{N}
\bigg\{
(1+c_b) B^{(k)}_t+c_b\Big[B^{(k)}_{t+1} +\sum_{|u|>0} b^{-\frac{k-1-t-|u|}{2}}\sum_{|\kappa|\leq k-1-t-|u|}\,b^{\frac{|\kappa|}{2}}\sigma_{u,\kappa}^2\Big]\bigg\}\\
&+b^{2t}\frac{\sigma^2}{N^2}
\end{align*}
where $c_b=\frac{(b-1)^{1/2}}{b^{1/2}-1}$,
$$
B^{(k)}_c=\sum_{|u|>0}\, \sum_{|\kappa|>k-c-|u|}\sigma_{u,\kappa}^2+\sum_{|u|>0} b^{-(k-c-|u|)} \sum_{|\kappa|\leq k-c-|u|} \sigma_{u,\kappa}^2 b^{|\kappa|},\quad c\in\mathbb{N}
$$
and where we use the convention that empty sums are null.
\end{thm}
 
The bound provided in Theorem \ref{thm:GenOwen} is hard to interpret but its main purpose  is to study the rate at which the variance goes to zero as the quadrature size increases. Thanks to Kronecker's lemma,  we show in Corollary \ref{cor:smallN} below that  this theorem  implies that for any square integrable function the error  is of size $\smallo_P(N^{-1/2})$ without any restriction on $N$. Due to its importance for this work, Kronecker's lemma  is recalled  in Lemma \ref{lemma:Kronecker} below \citep[see, e.g.,][Lemma 2, p.390, for a proof]{Probability}.

\begin{lem}[Kronecker's Lemma]\label{lemma:Kronecker}
Let $(d_n)_{n\geq 1}$ be a sequence of positive increasing numbers such that $d_n\rightarrow \infty$ as $n\rightarrow\infty$, and let $(z_n)_{n\geq 1}$ be a sequence of numbers such that $\sum_{n=1}^{\infty}z_n$ converges. Then, as $N\rightarrow\infty$, $d^{-1}_N\sum_{n=1}^N d_nz_n\rightarrow 0$.
\end{lem}


If the expression of the bound given Theorem \ref{thm:GenOwen}  is rather complicated, we note from the proof of this result that the variance of quadratures based on points taken from  scrambled $(0,s)$-sequences is never larger than a universal constant $c_0^*$  times the Monte Carlo variance. In addition, for $t>0$,  we derive from the proof of this theorem a simple bound for the variance  which is in most cases sharper than the one given in \eqref{eq:basic}. These  results are collected in the following corollary.

\begin{cor}\label{cor:smallN}
Consider the set-up of Theorem \ref{thm:GenOwen}. Then, $$\var\big(I(\tilde{P}^{N},\varphi)\big)=\smallo(N^{-1}).
$$
In addition, 
\begin{align}
\var\big(I(\tilde{P}^{N},\varphi)\big)&\leq \frac{\sigma^2}{N}\bigg\{\Big[\Gamma_{t,s}^{(b)}(1+2c_b)\Big]^{1/2}+\frac{b^t}{N^{1/2}}\bigg\}^2\label{eq:B1}
\end{align}
and, for $t=0$,
\begin{equation}\label{eq:B2}
\var\big(I(\tilde{P}^{N},\varphi)\big)\leq \frac{\sigma^2}{N}e\,(3+2\sqrt{2})<15.85 \frac{\sigma^2}{N}.
\end{equation}
\end{cor}

\begin{proof}

To prove the error rate, let $\tilde{\sigma}_{u,l}^2=\sum_{\kappa:|\kappa|=l}\sigma_{u,\kappa}^2$ for $l\in\mathbb{N}$ and  note that, for any fixed integers $a>0$, $0\leq c<k-s$ and $u\subseteq\mathcal{S}$,
\begin{align*}
b^{-a(k-c-|u|)} \sum_{|\kappa|\leq k-c-|u|} \sigma_{u,\kappa}^2 b^{a|\kappa|}&=b^{-a(k-c-|u|)} \sum_{l=0}^{k-c-|u|} b^{al}\,\tilde{\sigma}_{u,l}^2 =b^{-a(k-c-|u|+1)} \sum_{l=1}^{k-c-|u|+1} b^{al}\,\tilde{\sigma}_{u,l-1}^2
\end{align*}
which converges to zero by Kronecker's lemma. Also, because $\sum_{|u|>0}\sum_{\kappa}\sigma_{u,\kappa}^2=\sigma^2$, this shows that, as $N\rightarrow \infty$, $B^{(k)}_{t}\cvz $ and $B_{t+1}^{(k)}\cvz $ and therefore, using Theorem \ref{thm:GenOwen}, as $N\rightarrow\infty$, $N\var\big(I(\tilde{P}^{N},\varphi)\big)\cvz$. The proof of the bounds given in \eqref{eq:B1}-\eqref{eq:B2} is postponed to  \ref{app:B}.
\end{proof}

When $N\leq b^t$, the  trivial bound $\var\big(I(\tilde{P}^{N},\varphi)\big)\leq\sigma^2$ is sharper than the bounds given in \eqref{eq:basic} and in \eqref{eq:B1}. To compare these latter when $N> b^t$, note that, for all $b\geq 2$, $(1+2c_b)<b(b+1)/(b-1)$. Thus, the bound in  \eqref{eq:B1}  is sharper that the one provided in \eqref{eq:basic}  for any $N>b^t N_{s}^{(b)}$ with
\begin{align}\label{eq:NS}
 N_{s}^{(b)}=\bigg\{\bigg(\frac{b+1}{b-1}\bigg)^s\bigg[\bigg(b\frac{b+1}{b-1}\bigg)^{1/2}-(1+2c_b)^{1/2}\bigg]^2\bigg\}^{-1}.
\end{align}
Simple computations show that $N_{s}^{(b)}$ decreases as $b\geq 2$  and/or  $s\geq 1$ increases and,  for $b\geq 5$,  $N_{s}^{(b)}<1$ for all $s\geq 2$.

Table \ref{Table:Nb} below gives the value of $\max(1,N_{s}^{(b)})$ for different prime numbers $b\geq 2$ and for dimension $s>b$ so that a $(0,s)$-sequence in base $b$ does not exist \citep[see][Corollary 4.36, p.141]{dick2010digital}; that is, for values of $s$ such that the bound given in \eqref{eq:B2} cannot apply. As one may expect, remark that the bound in \eqref{eq:B1} is larger than the one given in \eqref{target} for quadratures based on $(\lambda,t,m,s)$-nets.

\begin{table}[H]
\centering
\begin{tabular}{c|cccccc}
$s$&3&4&5&6&7&$>8$\\
\hline
$\max(1,N_{s}^{(2)})$&29.77&9.93&3.31&1.11&1&1\\
$\max(1,N_{s}^{(3)})$&-&1.05&1&1&1&1\\
$\max(N_{s}^{(5)},1)$&-&-&-&1&1&1\\
$\max(N_{s}^{(>7)},1)$&-&-&-&-&-&1\\
\end{tabular}
\caption{Value of $\max(N_{s}^{(b)},1)$ for different prime numbers  $b\geq 2$ and for $s>b$, with $N_s^{(b)}$ defined in \eqref{eq:NS}.\label{Table:Nb}}
\end{table}

Finally it is worth mentioning that the $\smallo_P(N^{-1/2})$ convergence rate for quadratures based on scrambled nets of arbitrary size was simultaneously established by Art B. Owen (personal communication) using a more direct proof. Nevertheless,  the bound given in Theorem \ref{thm:GenOwen} also allows to study situations where the integrand  depends on the size of the quadrature rule $N$, as explained in  Section \ref{sec:size}.

\subsection{Error rate for smooth integrands}\label{subsec_smooth}

In a recent paper, \citet[][Theorem 2]{Owen2014}  established that the best possible  rate for the variance we can have uniformly on $N$ is  $N^{-2}$. In this subsection we show that,  under some smoothness assumptions on $\varphi$,  this optimal rate is achieved by scrambled net quadrature rules  in the sense that there exist constants $\underline{c}<\bar{c}<\infty$ such that, for $N$ large enough,
$$
\underline{c}< N^2\var\left(I(\tilde{P}^N,\varphi)\right)< \bar{c}.
$$
More precisely, we focus  on functions $\varphi\in L_2\ui^s$  such that $\sigma_{u,\kappa}^2=\bigO(b^{-2|\kappa|})$ for all $u\subseteq\mathcal{S}$. Note that this condition is fulfilled when, e.g.,  $\varphi$ has continuous mixed partial derivative of order $s$ \citep[][Lemma 2]{Owen2008} or when $\varphi$ satisfies the generalized Lipschitz condition considered in \citet{Yue1999a}. 

For such integrands $\varphi$, simple computations yield the following result:

\begin{prop}\label{prop}
Consider the set-up of Theorem \ref{thm:GenOwen} and assume that $\sigma_{u,\kappa}^2=\bigO(b^{-2|\kappa|})$ for all $u\subseteq\mathcal{S}$. Then,
$$
\var\left(I(\tilde{P}^N,\varphi)\right)=\bigO(N^{-2}).
$$
\end{prop}

\begin{proof}[Proof of Proposition \ref{prop}.] 

First, note that under the assumption of the proposition,  \citet[][Theorem 2]{Owen1998} and \citet[][Theorem 3]{Owen2008} show that, for $m\geq t+s-1$,
\begin{equation}\label{eq:YM}
\sum_{|u|>0}\,\sum_{|\kappa|>m-t-|u|}\sigma_{u,\kappa}^2=\bigO( b^{-2m}m^{s-1}).
\end{equation}
Then, let $N\geq b^{t+s-1}$ and $k$ be the largest integer such that $N\geq b^k$.  The standard way to analyse the variance of a scrambled net  quadrature rule of arbitrary size  is to decompose $\tilde{P}^N$ into scrambled $(a_m,t,m,s)$-nets $\tilde{P}_m$, $m=t,\dots,k$, and a remaining set $\tilde{P}$ that contains $\tilde{n}<b^t$ points (see the proof of Theorem \ref{thm:GenOwen} for more details). Let $\tilde{P}'=\tilde{P}\cup_{m=t}^{t+s-2}\tilde{P}_m$. Then, using trivial inequalities and the convention that empty sums are null, we have
$$
\var\bigg(I(\tilde{P}^N,\varphi)\bigg)\leq \frac{1}{N^2}\bigg(\Big\{\var\Big(\sum_{\tilde{\bx}^n\in \tilde{P}'}\varphi(\tilde{\bx}^n)\Big)\Big\}^{1/2}
+\sum_{m=t+s-1}^k\Big\{\var\Big(\sum_{\tilde{\bx}^n\in \tilde{P}_m}\varphi(\tilde{\bx}^n)\Big)\Big\}^{1/2}
\bigg)^{2}.
$$
Let  $\tilde{I}\subset\{0,\dots,N-1\}$ be such that $n\in \tilde{I}$ if and only if $\tilde{\bx}^n\in \tilde{P}$. Then,
$$
\var\Big(\sum_{\tilde{\bx}^n\in \tilde{P}'}\varphi(\tilde{\bx}^n)\Big)\leq \Big(\sum_{n\in \tilde{I}}\Big\{\var\Big(\varphi(\tilde{\bx}^n)\Big)\Big\}^{1/2}\Big)^2=|\tilde{P}'|^2\sigma^2< (b^{t+s-1})^2\sigma^2.
$$
In addition,  using \eqref{res} and \eqref{eq:YM},
\begin{align*}
\sum_{m=t+s-1}^{k}\Big\{\var\Big(\sum_{\tilde{\bx}^n\in \tilde{P}_m}\varphi(\tilde{\bx}^n)\Big)\Big\}^{1/2}&=\bigO\Big(\sum_{m=t+s-1}^k b^{-m/2}m^{\frac{s-1}{2}}\Big).
\end{align*}
To conclude the proof, note that the serie $\sum_{m=t+s-1}^{\infty} b^{-m/2}m^{\frac{s-1}{2}}$ is convergent.
\end{proof}

We conclude this subsection with two remarks. First, and as in \citet[][Theorem 2]{Owen1998}, the computations in the proof of Proposition \ref{prop} hold for $N\geq b^{t+s}$ and thus we cannot expect that the variance decreases as $N^{-2}$   for smaller quadrature sizes. Second, the rate of order $\bigO_P(N^{-2}(\log N)^{s-1})$ found by \citet{Yue1999a}, under the same assumptions as in Proposition \ref{prop}, is due to the fact that, in the last step of the proof of this latter, they use the inequality $\sum_{m=t+s-1}^k b^{-m/2}m^{(s-1)/{2}}<k^{(s-1)/{2}}(1-b^{-1/2})^{-1}$ rather than using the fact that the series  $\sum_{m=t+s-1}^{\infty} b^{-m/2}m^{(s-1)/{2}}$ is convergent.

\section{Error rate for integrands that depend on the quadrature size}\label{sec:size}

We now analyse the behaviour of the quadrature $I (\tilde{P}^N,\varphi_N )$ where $(\varphi_N)_{N\geq 1}$ is a sequence of real valued functions. In practice, the sequence  of functions $(\varphi_N)_{N\geq 1}$ is often such that, as $N\rightarrow\infty$, $\varphi_N\rightarrow \varphi$   where $I(\varphi)$ is the quantity of interest. The classical situation where this set-up  occurs is when we are estimating $I(\varphi)$ using a sequential method such as the array-RQMC algorithm developed by \citet{LEcuyer2006} or  the SQMC algorithm proposed by \citet{SQMC}. 

Using Theorem \ref{thm:GenOwen}, we can deduce the following result concerning the error size of the quadrature rule $I(\tilde{P}^N,\varphi_N)$.

\begin{cor}\label{cor:phiN}
Consider the set-up of Theorem \ref{thm:GenOwen}. Let $(\varphi_N)_{N\geq 1}$ be a sequence of  functions such that, $\forall N\in\mathbb{N}^*$, $\varphi_N\in L^2\ui^s$, and  for $N\geq 1$, let 
$$
\sigma_N^2=\int_{\ui^s}\Big(\varphi_N(\bx)-\int_{\ui^s}\varphi_{N}(\mathbf{v})\dd\mathbf{v}\Big)^2\dx,\quad \sigma_{N}^2=\sum_{|u|>0}\,\sum_{\kappa}\sigma^2_{N,u,\kappa}.
$$
Assume that, for any $u\subseteq\mathcal{S}$ and for any $\kappa(u)$, we have, as $N\rightarrow \infty$,  $\sigma^2_{N,u,\kappa}\rightarrow \sigma^2_{u,\kappa}$  and $\sigma_{N}^2\rightarrow \sigma^2<\infty$, where $\sigma^2=\sum_{|u|>0}\,\sum_{\kappa}\sigma^2_{u,\kappa}$. Then,
$$
\var\big(I(\tilde{P}^{N},\varphi_N)\big)=\smallo(N^{-1}).
$$

\end{cor}
\begin{proof}
Let $k\in\mathbb{N}$ be the largest power of $b$ such that $b^k\leq N$. Then, by Theorem \ref{thm:GenOwen}, to prove the result we first need to show that, for  $a\in\{\frac{1}{2},1\}$ and $c\in\{t,t+1\}$, we have
$$
\sum_{|u|>0}b^{-a(k-c-|u|)} \sum_{|\kappa|\leq k-c-|u|} \sigma_{N,u,\kappa}^2 b^{a|\kappa|}=\smallo(1).
$$
To establish this result, let $a$ and $c$ be as above, $k\geq t+s+2$, $k'=k-c-|u|+1\geq2$, $\tilde{k}=\lfloor k'/2\rfloor\geq 1$ and $S_{u,p}^N=\sum_{l=1}^{p+1}\tilde{\sigma}_{N,u,l-1}^2$ where $\tilde{\sigma}^2_{N,u,l}$ is defined as in the proof of Corollary \ref{cor:smallN}. Note that the positive and increasing sequence $(S_{u,p}^N)_{p\geq 1}$  converges to $\sigma_{N,u}^2=\sum_{|\kappa|>0}\sigma^2_{N,u,\kappa}$ as $p\rightarrow  \infty$. Then, using summation by part and similar computations as in the proof of Kronecker's lemma \citep[see, e.g.,][Lemma 2, p.390]{Probability}, we have
\begin{align*}
b^{-a(k-c-|u|)} \sum_{l=0}^{k-c-|u|}\tilde{\sigma}_{N,u,l}^2 b^{al}&=b^{-a k'} \sum_{l=1}^{k'}\tilde{\sigma}_{N,u,l-1}^2 b^{al}\\
&=S^N_{u,k'} - b^{-a k'}\sum_{l=1}^{\tilde{k}-1}(b^{a(l+1)}-b^{al})S^{N}_{u,l}
- b^{-ak'}\sum_{l=\tilde{k}}^{k'-1}(b^{a(l+1)}-b^{al})\sigma_{N,u}^2\\
&-b^{-ak'} \sum_{l=\tilde{k}}^{k'-1} (b^{a(l+1)}-b^{al})(S^{N}_{u,l}-\sigma_{N,u}^2)\\
&\leq \frac{b^{a\tilde{k}}}{b^{ak'}}\sigma_{N,u}^2+ (\sigma_{N,u}^2-S^{N}_{u,\tilde{k}})
\end{align*}
so that (recall that $\tilde{k}$ and $k'$ depend on $|u|$)
\begin{align}\label{eq:InN}
\sum_{|u|>0}b^{-a(k-c-|u|)} \sum_{|\kappa|\leq k-c-|u|} \sigma_{N,u,\kappa}^2 b^{a|\kappa|}\leq  \sum_{|u|>0}b^{a(\tilde{k}-k')}\sigma_{N,u}^2+ \Big(\sigma_{N}^2-\sum_{|u|>0}S^{N}_{u,\tilde{k}}\Big).
\end{align}
Then, using Fatou's Lemma, 
\begin{align*}
0\leq \limsup_{k\rightarrow \infty}\Big(\sigma_{N}^2-\sum_{|u|>0}S^{N}_{u,\tilde{k}}\Big)&=\limsup_{k\rightarrow \infty} \Big(\sigma_{N}^2-\sum_{|u|>0}\sum_{l\geq 0}\mathbb{I}(l\leq \tilde{k})\tilde{\sigma}^2_{N,u,l}\Big)\\
&\leq \sigma^2-\liminf_{k\rightarrow \infty} \sum_{|u|>0}\sum_{l\geq 0}\mathbb{I}(l\leq \tilde{k})\tilde{\sigma}^2_{N,u,l}\\
&\leq \sigma^2- \sum_{|u|>0}\sum_{l\geq 0}\liminf_{k\rightarrow \infty}\mathbb{I}(l\leq \tilde{k})\tilde{\sigma}^2_{N,u,l}\\
&=0
\end{align*}
because each $\tilde{\sigma}^2_{N,u,l}$ is a finite sum of some $\sigma^2_{N,u,\kappa}$'s and, by assumption, $\sigma^2_{N,u,\kappa}\rightarrow \sigma^2_{u,\kappa}$ for any $u$ and $\kappa$. This shows that the second term of \eqref{eq:InN} converges to zero as $N\rightarrow \infty$. The above computations also show that, for any $u\subseteq\mathcal{S}$, $\sigma^2_{N,u}$ converges to $\sum_{|\kappa|>0}\sigma^2_{u,\kappa}$  so that $b^{a(\tilde{k}-k')}\sigma_{N,u}^2\rightarrow 0$ as $N\rightarrow \infty$. Hence, the right-hand side of \eqref{eq:InN} goes to zero as $N$ increases, as required. To conclude the proof note that these  computations also imply that, as $N\rightarrow\infty$, $\sum_{|u|>0}\sum_{|\kappa|\leq k-c-|u|} \sigma_{N,u,\kappa}^2\cvz$.

\end{proof}

\subsection{Application of Corollary \ref{cor:phiN} to SQMC and to sampling importance resampling}

A direct consequence of Corollary \ref{cor:phiN} is to relax the constraint on $N$ in  \citet[][Theorem 7]{SQMC}, showing that on the class of continuous and bounded functions SQMC asymptotically outperforms standard sequential Monte Carlo algorithms without any restriction on how the number of simulations (or ``particles'')  grows. 

Providing a complete description of SQMC is beyond the scope of this work (see however Section  \ref{sec:exSS} for an example of SQMC algorithm). Nevertheless, to get some insight about how  Corollary \ref{cor:phiN}  applies to this class of methods, we illustrate this result by studying  a scrambled net version of the  sampling importance resampling (SIR) algorithm proposed by \citet{Rubin1987,Rubin:SIR}, which is iteratively used in SQMC. In addition,  and as already mentioned, the steps used to establish the convergence rate of the latter (Proposition \ref{prop:SIR} below) are exactly the same that those needed to extend \citet[][Theorem 7]{SQMC} to an arbitrary pattern of $N$.


SIR algorithms are designed to estimate the expectation
$\pi(f):=\int_{\ui^s} f(\bx)\pi(\bx)\dx$,
with $\pi$  a density function on $\ui^s$; see  Algorithm \ref{algo:SIR}  for the pseudo-code version of the proposed QMC version of SIR. In Algorithm \ref{algo:SIR},  $q(\bx)\dx$ is a proposal distribution on $\ui^s$ and, for a probability measure $\mu$ on $\ui^s$, $\bm{F}^{-1}_{\mu}:\ui^s\rightarrow\ui^s$  denotes the (generalized) inverse of $\bm{F}_{\mu}$, the Rosenblatt transformation of $\mu$ \citep[see][for a definition]{Rosenblatt1952}. Finally, $h: [0,1]^s\rightarrow [0,1]$ is a pseudo-inverse of the Hilbert space filling curve $H:[0,1]\rightarrow [0,1]^s$, which is a continuous  mapping from the unit interval onto the unit hypercube \citep[see, e.g.,][for how to construct the Hilbert curve for any $s\geq 2$]{Hamilton2008b}.

Using Corollary \ref{cor:phiN} and the results in \citet{SQMC}, we can prove that the error of Algorithm \ref{algo:SIR} to approximate $\pi(f)$ is of size $\smallo_P(N^{-1/2})$ for any pattern of $N$, as shown in the next result.

\begin{algorithm}
\caption{QMC sampling importance resampling \label{algo:SIR}}
\begin{algorithmic}[1]
\State Generate $\tilde{P}_1^N=\{\tilde{\bx}_1^n\}_{n=0}^{N-1}$ a scrambled net in $\ui^s$ and $\tilde{P}_2^N=\{\tilde{x}_2^n\}_{n=0}^{N-1}$ a scrambled net in $\ui$
\State Compute $\bz_1^n=\bm{F}_q^{-1}(\tilde{\bx}_1^n)$ and $w^n=\pi(\bz_1^n)/q(\bz_1^n)$ for $n=0,\dots,N-1$
\State\label{Step2}Compute $I(\tilde{P}_2^N,\varphi_N)$ where $\varphi_N=f\circ F_{h,N}^{-1}$ with $F^{-1}_{h,N}$ the (generalized) inverse of the empirical cumulative density function $F_{h,N}(z)=\sum_{n=0}^{N-1} \frac{w^n}{\sum_{m=0}^{N-1} w^m}\mathbb{I}(h(\bz_1^n)\leq z)$
\State \textbf{return} $I(\tilde{P}_2^N,\varphi_N)$, an estimate of $\pi(f)$
\end{algorithmic}
\end{algorithm}

\begin{prop}\label{prop:SIR}
Consider Algorithm \ref{algo:SIR} where $\tilde{P}_1$ and $\tilde{P}_2$ are independent. Assume that the functions $f(\bz)$ and $\pi(\bz)/q(\bz)$ are continuous and bounded on $\ui^s$  and that, for all $i\in\mathcal{S}$, the $i$-th component of $\bm{F}_q^{-1}$ is continuous on $\ui^i$. Then, $\varphi_N\rightarrow \varphi:=f\circ F^{-1}_{\pi_h}$ almost surely,  with $\pi_h$  the image by $h$ of $\pi$ and with $I(\varphi)=\pi(f)$. In addition,
$$
\var\Big(I(\tilde{P}_2^N,\varphi_N)\Big)=\smallo_P(N^{-1}),\quad \E|I(\tilde{P}_2^N,\varphi_N)-\pi(f)|=\smallo_P(N^{-1/2}).
$$
\end{prop}

\begin{proof}
To show that $\varphi_N\rightarrow \varphi$ almost surely, note first that, as $N\rightarrow \infty$, $\|F_{h,N}-F_{\pi_h}\|_{\infty}\cvz$  almost surely by \citet[][Theorem 1 and Theorem 3]{SQMC} and therefore, for all $x\in\ui$, $F_{h,N}^{-1}(x)\rightarrow F^{-1}_{\pi_h}(x)$ with probability one \citep[see the proof of][Theorem 7]{SQMC}. Then, $\varphi_N\rightarrow\varphi$ because $f$ is continuous.

To show the second part of the proposition, let $\mathcal{F}_1^N$ be the $\sigma$-algebra generated by $\{\bz_1^n\}_{n=0}^{N-1}$ and note that
\begin{align}
\var\Big(I(\tilde{P}_2^N,\varphi_N)\Big)&=\E\left[\var\Big(I(\tilde{P}_2^N,\varphi_N)\big|\mathcal{F}_1^N\Big)\right]+\var\left(\E\Big[I(\tilde{P}_2^N,\varphi_N)\big|\mathcal{F}_1^N\Big]\right)\notag\\
&=\E\left[\var\Big(I(\tilde{P}_2^N,\varphi_N)\big|\mathcal{F}_1^N\Big)\right]+\var\left(\pi_N(f)\right)\label{eq:intSIR}
\end{align}
where $\pi_N(f)=\sum_{n=0}^{N-1}w^nf(\bz_1^n)/(\sum_{m=0}^{N-1} w^m)$. 

The second term after the second equality sign is $\smallo(N^{-1})$ by  Corollary \ref{cor:smallN} and by \citet[][Theorem 2]{SQMC}. For the first term, the same  computations as in the proof of \citet[][Theorem 7]{SQMC} show that the sequence $(\varphi_N)_{N\geq 1}$ verifies with probability one the assumptions of  Corollary \ref{cor:phiN}. Thus, using this latter,
$$
\var\Big(I(\tilde{P}_2^N,\varphi_N)\big|\mathcal{F}_1^N\Big)=\smallo(N^{-1}),\quad\text{almost surely}
$$ 
and therefore the first  term in \eqref{eq:intSIR} is $\smallo(N^{-1})$ by the Dominated Convergence theorem. This shows that $\var\Big(I(\tilde{P}_2^N,\varphi_N)\Big)=\smallo(N^{-1})$. Finally, to establish the result for the $L_1$-norm, note that 
\begin{align*}
 \E|I(\tilde{P}_2^N,\varphi_N)-\pi(f)|&\leq  \E|I(\tilde{P}_2^N,\varphi_N)-\pi_N(f)|+ \E|\pi_N(f)-\pi(f)|\\
&\leq  \left\{\var\Big(I(\tilde{P}_2^N,\varphi_N)\Big)\right\}^{1/2}+ \E|\pi_N(f)-\pi(f)|
\end{align*}
where the first term after the second inequality sign is $\smallo(N^{-1/2})$ from the above computations while the second term is $\smallo(N^{-1/2})$ by Corollary \ref{cor:smallN} and by \citet[][Theorem 2]{SQMC}.
\end{proof}

\section{Numerical Study}\label{sec:num}

In this section we illustrate the main findings of this paper. All the simulations presented below rely on a Sobol' sequence that is scrambled using the method proposed by \citet{Owen1995}. We recall  that $b=2$ for the Sobol' sequence.  

\subsection{Scrambled  net quadrature rules}\label{sec:num1}

We  consider the problem of estimating the $s$-dimensional integral $I(\varphi_j)$, $j=1,\dots,4$, where
\begin{align*}
\varphi_{1}(\bx)=\sum_{i=1}^s x_i,\quad \varphi_{2}(\bx)=\max\Big(\sum_{i=1}^sx_i-\frac{s}{2}, 0\Big),\quad \varphi_{3}(\bx)=\mathbb{I}_{\left(\sum_{i=1}^s x_i>\frac{s}{2}\right)}(\bx)
\end{align*}
are as in \citet{He2014} and where $\varphi_{4}(\bx)=12^{s/2}\prod_{i=1}^s(x_i-0.5)$ is as in \citet{Owen1997b,Owen1998}. Note that the integrands $\varphi_{1}$ and $\varphi_{2}$ are both Lipschitz continuous but $\varphi_{2}$ is not everywhere differentiable, while $\varphi_{4}$ satisfies the assumption of Proposition \ref{prop} \citep{Owen1997b}.  For $j\in\{1,\dots, 4\}$, we estimate the integral $I(\varphi_{j})$  using  the quadrature rule $I(\tilde{P}^N,\varphi_{j})$ where, as mentioned above, $\tilde{P}^N$ is the set containing the first $N$ points of a scrambled Sobol' sequence. 

Figure \ref{Fig:num1} shows  the evolution of the  mean square errors (MSEs) as a function of $N$. Results are presented for $N$  ranging from $1$ to $2^{18}$, with $s=3$ for $\varphi_j$, $j=1,2,3$, and $s=6$ for $\varphi_4$. In addition to the MSEs, we have reported the Monte Carlo $N^{-1}$ reference line to illustrate the result of Corollary \ref{cor:smallN}, namely that the convergence rate is faster than $N^{-1}$ for any pattern of $N$. To illustrate the finding of Proposition \ref{prop}, we have also represented a $N^{-2}$ reference line in the plot showing the results for the quadrature $I(\tilde{P}^N,\varphi_{4})$ (Figure \ref{Fig:phi4}). This $N^{-2}$ reference line starts at $N=b^{t+s}=2^{14}$, which is the value of $N$ from which we can naively  expect that the quadrature $I(\tilde{P}^N,\varphi_4)$ enters in the asymptotic regime (see Section \ref{subsec_smooth}).



To compare quadrature rules based on nets of arbitrary size with those based on $(t,m,s)$-nets, Figure \ref{Fig:num1} also shows the evolution of the MSEs along the subsequence $N=2^m$. The interesting point to note here is that, for a given value of $s$, the advantage of using $(t,m,s)$-nets over nets of arbitrary size decreases as the integrand becomes ``less smooth''. Indeed, for the everywhere differentiable and Lipschitz function $\varphi_1$, we observe that taking for $N$ powers of 2 significantly improves the convergence rate. In addition, this choice for the quadrature size  $N$  is also the cheapest way to reach any given level of MSE. For the function $\varphi_2$ this observation holds for $N\geq 2^8$ but  the gain in term of convergence rate  is  smaller than for the estimation of $I(\varphi_1)$. Finally, the advantage of taking a power of 2 for the quadrature size has completely disappeared for the discontinuous function $\varphi_3$. 

To understand these observations recall that, by Proposition \ref{prop}, the  (asymptotic) convergence rate of  $\var\big(I(\tilde{P}^N,\varphi)\big)$ is  $N^{-2}$  uniformly on $N$ when $\varphi$ is  smooth enough so that the quantities $\sigma_{u,\kappa}^2$'s decrease sufficiently quickly as $|\kappa|$ increases. However, under the conditions of Proposition \ref{prop}, the error size of quadratures based on scrambled $(t,m,s)$-nets is of order $\bigO(N^{-3/2}(\log N)^{(s-1)/2}$ \citep[][Theorem 3]{Owen2008} and is thus smaller than what is obtained for an arbitrary value of $N$. More generally, and  as illustrated in Figure \ref{Fig:num1}, the error size of quadratures based on scrambled $(\lambda, t,m,s)$-nets depends positively on the smoothness of the integrand \citep[for more theoretical results on this point,  see][]{Owen1997b,Owen1998, Yue1999a, Hickernell2001}.  Consequently, taking $N=b^k$ is   the best choice for the smooth integrands $\varphi_1$ and $\varphi_2$   since then the MSE  goes to zero much faster than $N^{-2}$. Note that for $\varphi_2$ the MSE obtained by taking $N= 2^k$ decreases slower than for  $\varphi_1$ and, as a result, $N$ should be larger to rule out the choice $N\neq  b^k$. Finally, for the discontinuous function $\varphi_3$ the convergence rate of the MSE when using $(t,m,s)$-nets is too slow for the choice of  $N$ to influence that of the MSE.

\begin{figure}
\begin{center}
\begin{subfigure}{0.4\textwidth}
\begin{center}
\includegraphics[trim = 5cm 2cm 2.5cm 0,scale=0.3]{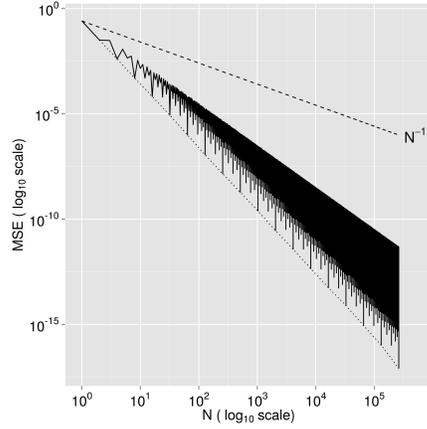}
\caption{\label{Fig:phi1}}
\end{center}
\end{subfigure}
\hspace{1.2cm}
\begin{subfigure}{0.4\textwidth}
\begin{center}
\includegraphics[trim = 5cm 2cm 2.5cm 0,scale=0.3]
{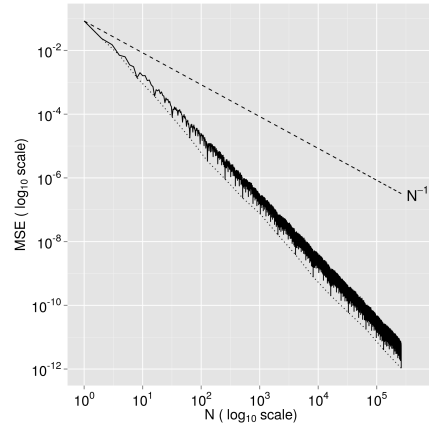}
\caption{\label{Fig:phi2}}
\end{center}
\end{subfigure}

\begin{subfigure}{0.4\textwidth}
\begin{center}
\includegraphics[trim = 5cm 2cm 2.5cm 0,scale=0.3]{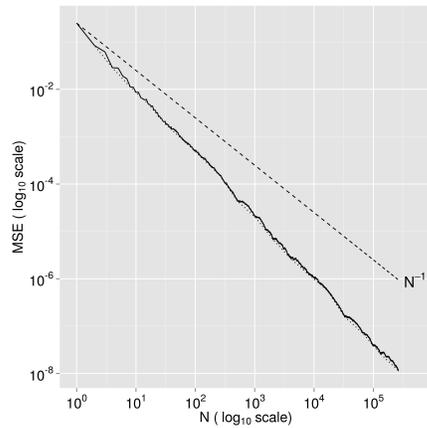}
\caption{\label{Fig:phi3}}
\end{center}
\end{subfigure}
\hspace{1.2cm}
\begin{subfigure}{0.4\textwidth}
\begin{center}
\includegraphics[trim = 5cm 2cm 2.5cm 0,scale=0.3]
{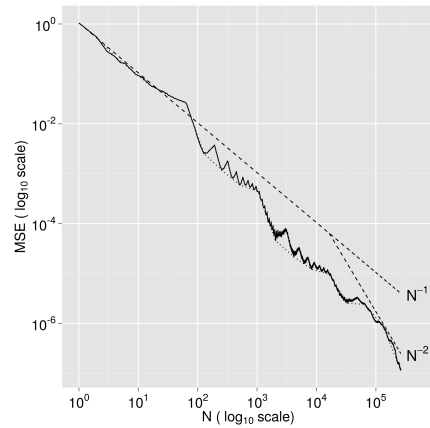}
\caption{\label{Fig:phi4}}
\end{center}
\end{subfigure}
\end{center}
\caption{Mean square error  of $I(\tilde{P}^N,\varphi_1)$ (Figure \ref{Fig:phi1}), $I(\tilde{P}^N,\varphi_2)$ (Figure \ref{Fig:phi2}), $I(\tilde{P}^N,\varphi_3)$ (Figure \ref{Fig:phi3}) and $I(\tilde{P}^N,\varphi_4)$  where $\tilde{P}^N$ contains the first $N$ points of a scrambled Sobol' sequence. The  dotted lines present the results along the subsequence $N=2^m$ for $m=0,\dots,18$ and the solid lines the MSEs for any $N\in\{1,\dots 2^{16}\}$. In Figure \ref{Fig:phi4}, the $N^{-2}$ reference line starts at $N=b^{t+s}=2^{14}$. The results are obtained from 1\,000 independent repetitions.\label{Fig:num1}}
\end{figure}

\subsection{Likelihood function estimation in state space models}\label{sec:exSS}

We now study the problem of estimating the likelihood function of the following generic univariate state space model
\begin{equation}\label{eq:SS}
\begin{cases}
y_k|z_k\sim\mathcal{N}\left(\mu_{y}(z_k),\sigma_{y}^2(z_k)\right),& k\geq 0\\
z_k|z_{k-1}\sim \mathcal{N}\left(\mu_{z}(z_{k-1}), \sigma_{z}^2(z_{k-1})\right),&k\geq 1\\
z_0\sim\mathcal{N}(\mu_0,\sigma_0^2)
\end{cases}
\end{equation}
where $(y_k)_{k\geq 0}$ is the observation process, $(z_k)_{k\geq 0}$ is the hidden Markov process and where $\mu_{q}:\mathbb{R}\rightarrow\mathbb{R}$ and $\sigma_{q}:\mathbb{R}\rightarrow\mathbb{R}^+$, $q\in\{z,y\}$, are known functions. 

Given a set of $T\geq 1$ observations $\{y_k\}_{k=0}^{T-1}$, we denote by $p(y_{0:T-1})$ the likelihood  function of the model defined by (\ref{eq:SS}), which cannot be computed explicitly.  Indeed, writing $f(\cdot{ },\mu,\sigma^2)$  the density function of the $\mathcal{N}(\mu,\sigma^2)$ distribution, it is easy to see that (using the convention that $f\big(z_k,\mu_z(z_{k-1}),\sigma^2_z(z_{k-1})\big)=f\big(z_0,\mu_0,\sigma_0^2\big)$ when $k=0$)
\begin{equation}\label{eq:intSS}
p(y_{0:T-1})=\int_{\mathbb{R}^T}\prod_{k=0}^{T-1} f\big(y_k,\mu_{y}(z_{k}),\sigma_{y}^2(z_{k})\big)f\big(z_k,\mu_z(z_{k-1}),\sigma^2_z(z_{k-1})\big)\dd z_k
\end{equation}
where, in practical scenarios, the time horizon $T$ is large (at least several dozen). In addition, simple unweighed  quadrature rules  are generally very inefficient to evaluate the integral appearing in \eqref{eq:intSS}. To see this, note that $p(y_{0:T-1})=I(\varphi_T)$ where $\varphi_T:\ui^T\rightarrow \mathbb{R}$ is given by 
$$
\varphi_T(x_0,\dots,x_{T-1})=\tilde{\varphi}_{T}\circ F^{-1}_{T}(x_0,\dots,x_{T-1})
$$
with $\tilde{\varphi}_{T}(z_0,\dots,z_{T-1})=\prod_{k=0}^{T-1}f(y_k,\mu_{y}(z_k),\sigma^2_{y}(z_k))$ and
$F_{T}$ the Rosenblatt transformation   of the probability measure on $\mathbb{R}^{T}$ defined by
$$
\prod_{k=0}^{T-1} f\big(z_k,\mu_{z}(z_{k-1}),\sigma_{z}^2(z_{k-1})\big)\dd z_k.
$$ 
Because $T$ is typically large, the function $\varphi_T$ is concentrated in a tiny region of the integration domain and, consequently,   quadrature rules  require a huge number of points to provide a precise estimate of $p(y_{0:T-1})$. An efficient way to get an (unbiased) estimate $p^N(y_{0:T-1})$ of $p(y_{0:T-1})$ is to use a SQMC algorithm \citep{SQMC}; that is, a QMC version of sequential Monte Carlo methods which are standard tools to handle this kind of problems  \citep[see, e.g.,][]{DouFreiGor}. The suitable SQMC algorithm for the generic state space model (\ref{eq:SS}) is presented in Algorithm \ref{alg:SQMC_SS}, where we use the standard notation  $\Phi(\cdot)$ for the cumulative density function (CDF) of the $\mathcal{N}(0,1)$ distribution. Note that inference in state space model (\ref{eq:SS}) is just an example of problems that can be addressed using SQMC and, in particular,  SQMC is not restricted to Gaussian models.

To see the connection between the results presented in Sections \ref{sec:res}-\ref{sec:size}  and  Algorithm \ref{alg:SQMC_SS}, note that the likelihood function $p(y_{0:T-1})$ can be decomposed as follows:
$$
p(y_{0:T-1})=\prod_{k=0}^{T-1}p(y_k|y_0,\dots,y_{k-1})
$$
with the convention that $p(y_k|y_0,\dots,y_{k-1})=p(y_0)$ when $k=0$.  Then, Algorithm \ref{alg:SQMC_SS} amounts to recursively computing an approximation of the form $N^{-1}\sum_{n=1}^Nw_k^n=I(\varphi_{N,k}, \tilde{P}_k^N)$
of the incremental likelihood $p(y_k|y_0,\dots,y_{k-1})$, $k=0,\dots,T-1$. 

At iteration $k=0$, $\varphi_{N,k}$ is the  function $f(y_0,\mu_y(\cdot),\sigma_y^2(\cdot))$ which therefore does not depend on $N$. Thus, iteration 0 of Algorithm \ref{alg:SQMC_SS} is a simple scramble net quadrature rule  which enters in the framework of Section \ref{sec:res}. For $k\geq 1$, it is easy to see that
$$
\varphi_{N,k}(\bx)=f\left(y_k,\mu_y\circ g(\bx),\sigma^2_y\circ g(\bx)\right),\quad g(\bx)=\mu_z\circ F_{N,k-1}^{-1}(x_1)+\sigma_z\circ F_{N,k-1}^{-1}(x_1)\Phi^{-1}(x_2)
$$
and thus, for $k\geq 1$, we are in the set-up of Section \ref{sec:size} where the integrand depends on the quadrature size $N$.

\begin{algorithm}
\caption{SQMC Algorithm to estimate $p(y_{0:T-1})$ in the state space model (\ref{eq:SS}) \label{alg:SQMC_SS}}
\begin{algorithmic}[1]
\State\label{state:RQMC1}Generate a RQMC point set $\tilde{P}^N_0=\{\tilde{x}_{0}^{n}\}_{n=0}^{N-1}$ in $\ui$
\State Compute $z_{0}^{n}=\mu_0+\sigma_0\Phi^{-1}(\tilde{x}_{0}^{n})$ and $w_0^n=f\left(y_0,\mu_{y}(z_0^n),\sigma^2_{y}(z_0^n)\right)$, $n=0,\dots, N-1$
\State Normalize the weights: $W_{0}^{n}=w_0^n/\sum_{m=0}^{N-1} w_0^m$, $n=0,\dots, N-1$

\State Compute $p^N(y_{0})=N^{-1}\sum_{n=0}^{N-1} w_0^n$
\For{$k=1\rightarrow T-1$}
\State\label{state:RQMC2} Generate a RQMC point set $\tilde{P}_k^N=\{\tilde{\bx}^n_{k}\}_{n=0}^{N-1}$  in $\ui^{2}$;
let $\tilde{\bx}_{k}^{n}=(\tilde{x}_{k}^{n},\tilde{v}^n_k)$
\For{$n=0\rightarrow N-1$}
\State\label{state:label} Compute $\tilde{z}_{k-1}^{n}=F_{N,k-1}^{-1}(\tilde{x}_k^n)$ where $F_{N,k-1}(z)=\sum_{m=0}^{N-1}W_{k-1}^m\mathbb{I}(z^m_{k-1}\leq z)$
\State Compute $z_{k}^{n}=\mu_{z}(\tilde{z}_{k-1}^{n})+\sigma_{z}
(\tilde{z}_{k-1}^{n})\Phi^{-1}(\tilde{v}_k^n)$
\State Compute $w_k^n=f\left(y_k,\mu_{y}(z_k^n),\sigma^2_{y}(z_k^n)\right)$
\EndFor
\State Normalize the weights: 
$W_{k}^{n}=w_k^n/\sum_{m=0}^{N-1} w_k^m$, $n=0,\dots,N-1$
\State Compute $p^{N}(y_{0:k})=p^N(y_{0:k-1})\,N^{-1}\sum_{n=0}^{N-1}
w_k^n$
\EndFor
\State \textbf{return} $p^N(y_{0:T-1})$, an estimate of $p(y_{0:T-1})$
\end{algorithmic}
\end{algorithm}

In this simulation study we analyse the MSE  of $\log p^N(y_{0:T-1})$ when at Step \ref{state:RQMC1} and at Step \ref{state:RQMC2} of Algorithm \ref{alg:SQMC_SS} the RQMC point sets are the first $N$ points of independent scrambled Sobol' sequences, where $N=4i$ for $i=3,\dots,2^{11}$. Note that, since the function $F_{N,k-1}^{-1}$ in the definition of $\varphi_{N,k}$ is discontinuous, the results of the previous subsection  suggest that the gain of restricting $N$ to be  powers of the Sobol' sequence can only be moderate in the context of SQMC. In addition, it is worth  remarking that this gain will also depend  on the regularity of the functions $\mu_{q}$ and $\sigma_{q}$, $q\in\{y,z\}$.

\subsubsection{Stochastic volatility (SV) model} 

We first consider the following simple univariate SV model
\begin{equation}\label{eq:SV}
\begin{cases}
y_k|z_k\sim\mathcal{N}\left(0,e^{-0.1+z_k}\right),& k\geq 0\\
z_k|z_{k-1}\sim \mathcal{N}\left(0.9z_{k-1},0.1\right),&k\geq 1\\
z_0\sim\mathcal{N}(0,\frac{0.1}{1-0.9^2})
\end{cases}
\end{equation}
from which a set of 100 observations is generated. Figure \ref{Fig:SV}  presents the MSE of the estimator  $\log p^N(y_{0:T-1})$ as well as the $N^{-1}$ Monte Carlo reference line. As expected from the results of Section \ref{sec:size}, we see that the $\smallo_P(N^{-1})$ convergence rate for the SQMC algorithm holds uniformly on $N$.  Nevertheless, we observe in this  example that selecting $N=2^m$ is optimal as soon as $N\geq 2^9$ in the sense that this choice  guarantees the smallest MSE for a given computational budget.

Interestingly, despite the discontinuities of the integrand we are facing at iteration $k\geq 1$ of SQMC, these first results look like those obtained for $\varphi_2$ in Section \ref{sec:num1} rather than like the ones we obtained for the discontinuous mapping $\varphi_3$. A possible explanation for this apparent contradiction is that the only source of discontinuities comes from the function $F_{N,k-1}^{-1}$. However, as $N$ increases, we can expect that the empirical CDF  of the particles generated at time $k\geq 0$ (and its inverse)  converges to a continuous function since  all the random variables  have a continuous distribution. Under some conditions we can show that this is indeed the case \citep[see, e.g., Proposition \ref{prop:SIR} above or the proof of Theorem 7 in][]{SQMC}. In addition, the SV model (\ref{eq:SV}) is an example of state space model  (\ref{eq:SS}) where the functions $\mu_{q}$ and $\sigma^2_{q}$, $q\in\{y,z\}$, are  very smooth.

\subsubsection{A non-linear and non-stationary model}

We now consider the following non-linear and non-stationary well known toy example in the particle filtering literature \citep[see, e.g.,][]{Gordon}
\begin{equation}\label{eq:NL}
\begin{cases}
y_k|z_k\sim\mathcal{N}\left(\frac{z_k^2}{20},1\right),& k\geq 0\\
z_k|z_{k-1}\sim \mathcal{N}\left(0.5z_{k-1}+25\frac{z_{k-1}}{1+z_{k-1}^2}+8\cos(1.2 k),10\right),&k\geq 1\\
z_0\sim\mathcal{N}(0,2)
\end{cases}
\end{equation}
from which we again simulate a set of 100 observations. Note that, in addition to the non-linearity of $\mu_z$, the density of the law of $y_k|x_k$ is bimodal when $y_k>0$. Due to these additional difficulties, we therefore expect that the gain of  restricting $N$ to be a power of 2  is less profitable than for the SV model. This point is confirmed in the  Figure \ref{Fig:NL}  where we show the evolution of the MSE as a function of $N$. We indeed remark from this plot that there is no gain of using a number of particles which is a power of two for the values of $N$ considered is this numerical study. However, and as for the SV model, we observe that SQMC converges faster the $N^{-1/2}$  Monte Carlo error rate.

To conclude this section it is worth  mentioning that to keep the presentation of SQMC simple we have only shown simulations for univariate models. In the multivariate version of SQMC, the resampling step of Algorithm \ref{alg:SQMC_SS} (Step \ref{state:label}) requires to sort the particles along a Hilbert space filling curve, as in Step \ref{Step2}  of the scrambled net SIR algorithm (Algorithm \ref{algo:SIR}). Since the Hilbert curve is $(1/d)$-H\"older continuous,  with $d$ the dimension of the state variable, the estimation problem becomes less smooth as $d$ increases. In light of the observations of this simulation study, this suggests that the gain of restricting $N$ to be  powers of the base of the underlying $(t,s)$-sequence is smaller than for univariate models. This point was confirmed in  non reported simulation study conducted for the bivariate version of the SV model (\ref{eq:SV}), where the gain of using $(t,m,s)$-nets  as input of SQMC has completely disappeared.

\begin{figure}
\centering
\begin{subfigure}{0.4\textwidth}
\begin{center}
\includegraphics[scale=0.35]{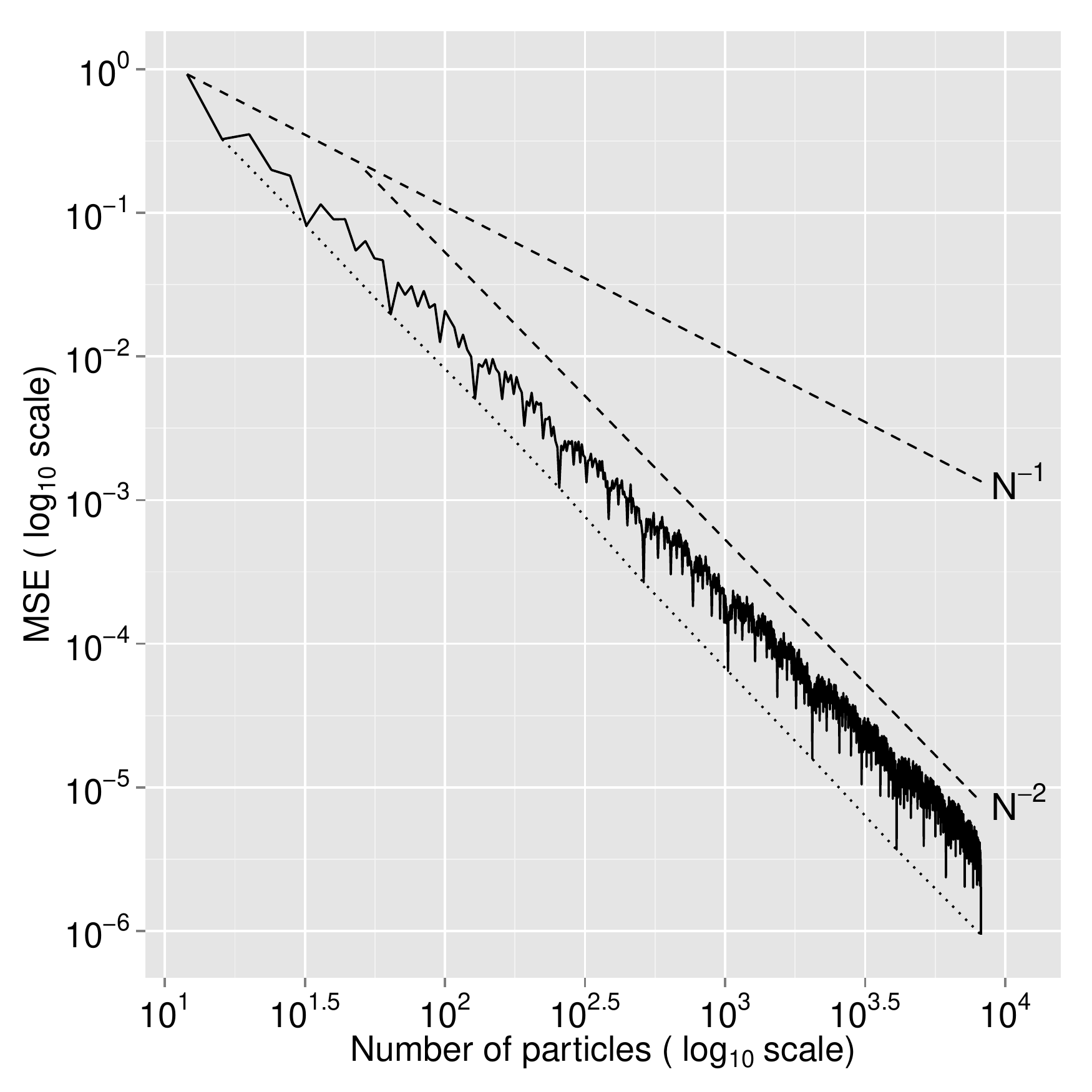}
\caption{\label{Fig:SV}}
\end{center}
\end{subfigure}
\hspace{1.2cm}
\begin{subfigure}{0.4\textwidth}
\begin{center}
\includegraphics[scale=0.35]{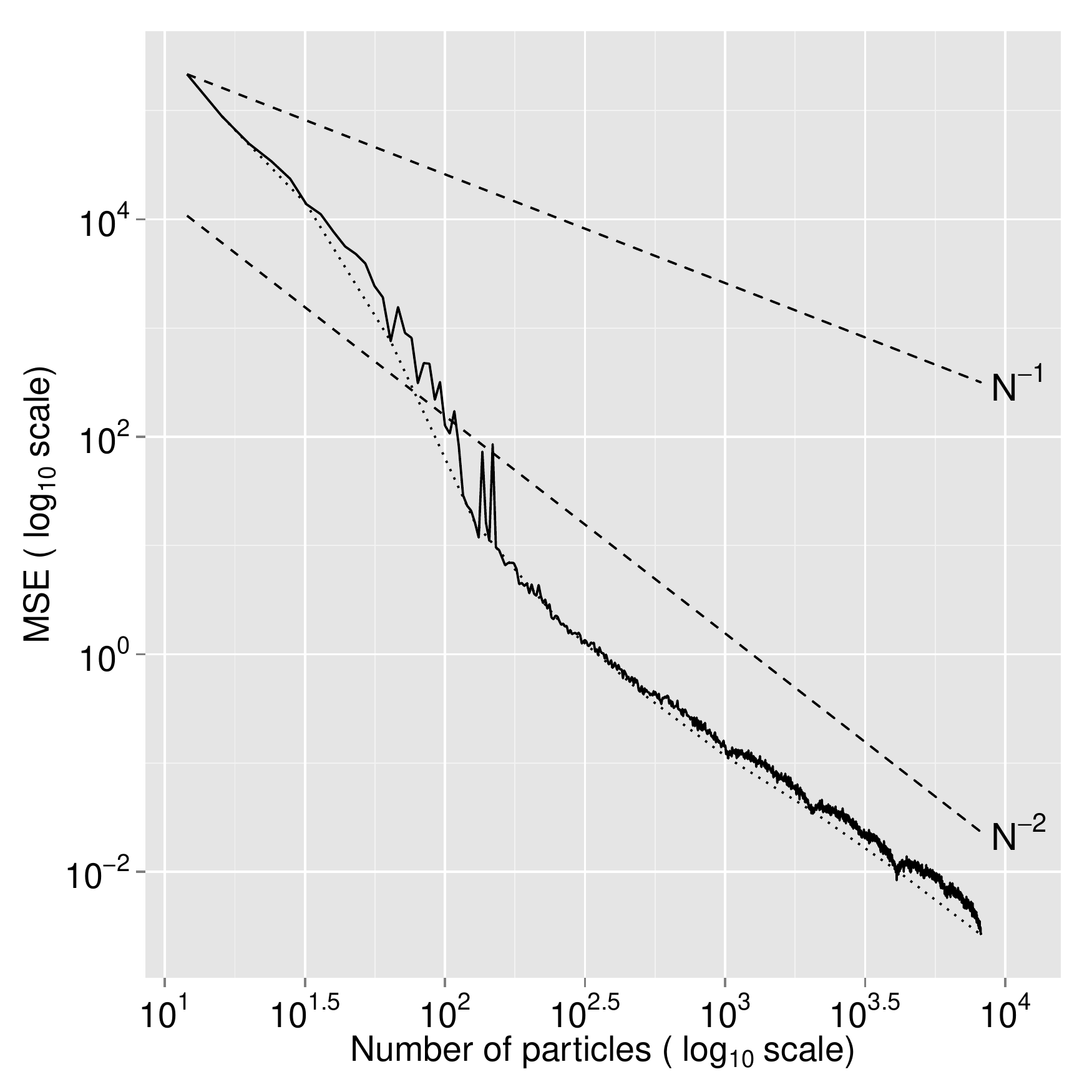}
\caption{\label{Fig:NL}}
\end{center}
\end{subfigure}
\caption{Mean square errors for the estimation of $\log p(y_{0:T-1})$ in the SV model (\ref{eq:SV}) (Figure \ref{Fig:SV}) and in the toy example (\ref{eq:NL}) (Figure \ref{Fig:NL}).  The dotted lines present the results for SQMC for $N=2^m$, $m=3,\dots,13$,  while the solid lines are for SQMC with $N=4i$, $i=3,\dots, 2^{11}$.   The results are obtained from 700 independent runs of Algorithm \ref{alg:SQMC_SS}.}
\end{figure}

\section{Conclusion}\label{sec:conc}

Together with the works of \citet{Yue1999a} and \citet{Hickernell2001}, the present analysis concludes to show that  the results of \citet{Owen1997a,Owen1997b,Owen1998} obtained for quadrature rules based on $(\lambda,t,s,m)$-nets are in fact true for quadrature rules based on the first $N$ points of scrambled $(t,s)$-sequences without any restriction on the pattern of $N$, namely, to sum-up:
\begin{enumerate}
\item For any square integrable functions the integration error goes to zero faster than for the classical Monte Carlo estimator;
\item\label{l2} For any square integrable functions the variance of scrambled quadrature rules is bounded by  the Monte Carlo variance multiplied by a constant  independent of the integrand;
\item The constant in \ref{l2}. is uniform with respect to the dimension for scrambled $(0,s)$-sequences;
\item For smooth integrands an explicit convergence rate (better than $N^{-1/2}$) can be computed \citep[see][]{Yue1999a, Hickernell2001}.
\end{enumerate} 
In a simulation study, we show that quadratures based on scrambled $(\lambda,t,m,s)$-nets outperform those based on nets of arbitrary size when the integrand $\varphi$ of interest is smooth. More precisely, using scrambled $(\lambda,t,m,s)$-nets is for such functions the fastest way to reach any given  level of MSE. Nevertheless, as  the integrand becomes less smooth, this gain decreases and completely disappears for discontinuous functions.

The second   important result proved in this paper is the asymptotic superiority of the sequential quasi-Monte Carlo algorithm proposed by \citet{SQMC} over standard sequential Monte Carlo methods without any restriction on how the number of particles grows. Since SQMC involves integration of discontinuous functions the behaviour of the MSE when the algorithm takes scrambled $(\lambda,t,m,s)$-nets as inputs should not be too different compared to what we would get when scrambled nets of arbitrary size are used. This point is illustrated in a simulation study based in two univariate state space models and we argue that for multivariate models it is very unlikely to expect any gain of using as input for SQMC only  points of scrambled sequences that form  $(\lambda,t,m,s)$-nets.

\section*{Acknowledgements}

I thank Nicolas Chopin, Art B. Owen, Florian Pelgrin and two anonymous referees for useful remarks that greatly  improve this paper. In addition, I am very grateful to Art B. Owen for having shared with me his  shorter proof for the first part of Corollary \ref{cor:smallN} he  derived when I was writing this manuscript.

\bibliographystyle{elsarticle-harv}
\bibliography{complete}

\begin{thebibliography}{23}
\expandafter\ifx\csname natexlab\endcsname\relax\def\natexlab#1{#1}\fi
\expandafter\ifx\csname url\endcsname\relax
  \def\url#1{\texttt{#1}}\fi
\expandafter\ifx\csname urlprefix\endcsname\relax\def\urlprefix{URL }\fi

\bibitem[{Andrieu et~al.(2010)Andrieu, Doucet, and Holenstein}]{PMCMC}
Andrieu, C., Doucet, A., Holenstein, R., 2010. Particle {M}arkov chain {M}onte
  {C}arlo methods. J. R. Statist. Soc. B 72~(3), 269--342.

\bibitem[{Dick and Pillichshammer(2010)}]{dick2010digital}
Dick, J., Pillichshammer, F., 2010. Digital {N}ets and {S}equences:
  {D}iscrepancy {T}heory and {Q}uasi-{M}onte {C}arlo {I}ntegration. Cambridge
  University Press.

\bibitem[{Doucet et~al.(2001)Doucet, de~Freitas, and Gordon}]{DouFreiGor}
Doucet, A., de~Freitas, N., Gordon, N.~J., 2001. Sequential {M}onte {C}arlo
  Methods in Practice. Springer-Verlag, New York.

\bibitem[{Doucet et~al.(2013)Doucet, Pitt, Deligiannidis, and
  Kohn}]{Doucet2013}
Doucet, A., Pitt, M., Deligiannidis, G., Kohn, R., 2013. Efficient
  implementation of {M}arkov chain {M}onte {C}arlo when using an unbiased
  likelihood estimator. arXiv preprint arXiv:1210.1871.

\bibitem[{Gerber and Chopin(2015)}]{SQMC}
Gerber, M., Chopin, N., 2015. Sequential {Q}uasi-{M}onte {C}arlo. J. R.
  Statist. Soc. B 77~(3), 509--579.

\bibitem[{Gordon et~al.(1993)Gordon, Salmond, and Smith}]{Gordon}
Gordon, N.~J., Salmond, D.~J., Smith, A. F.~M., 1993. Novel approach to
  nonlinear/non-{G}aussian {B}ayesian state estimation. IEE Proc. F, Comm.,
  Radar, Signal Proc. 140~(2), 107--113.

\bibitem[{Hamilton and Rau-Chaplin(2008)}]{Hamilton2008b}
Hamilton, C.~H., Rau-Chaplin, A., 2008. Compact {H}ilbert indices:
  Space-filling curves for domains with unequal side lengths. Inf. Process.
  Lett. 105~(5), 155--163.

\bibitem[{He and Owen(2014)}]{He2014}
He, Z., Owen, A.~B., 2014. Extensible grids: uniform sampling on a
  space-filling curve. arXiv:1406.4549.

\bibitem[{Hickernell and Yue(2001)}]{Hickernell2001}
Hickernell, F.~J., Yue, R.-X., 2001. The mean square discrepancy of scrambled
  $(t, s)$-sequences. SIAM J. Numer. Anal. 38, 1089--1112.

\bibitem[{L'Ecuyer et~al.(2006)L'Ecuyer, L\'{e}cot, and Tuffin}]{LEcuyer2006}
L'Ecuyer, P., L\'{e}cot, C., Tuffin, B., 2006. A randomized quasi-{M}onte
  {C}arlo simulation method for {M}arkov chains. In: Monte {C}arlo and
  {Q}uasi-{M}onte {C}arlo {M}ethods 2004. Springer Berlin Heidelberg, pp.
  331--342.

\bibitem[{Mato\v{u}sek(1998)}]{Matousek1998}
Mato\v{u}sek, J., 1998. On the ${L}_2$-discrepancy for anchored boxes. J.
  Complexity 14, 527--556.

\bibitem[{Niederreiter(1992)}]{Niederreiter1992}
Niederreiter, H., 1992. Random {N}umber {G}eneration and {Q}uasi-{M}onte
  {C}arlo {M}ethods. CBMS-NSF Regional conference series in applied
  mathematics.

\bibitem[{Owen(2008)}]{Owen2008}
Owen, A., 2008. Local antithetic sampling with scrambled nets. Ann. Statist.
  36~(5), 2319--2343.

\bibitem[{Owen(1995)}]{Owen1995}
Owen, A.~B., 1995. Randomly permuted $(t, m, s)$-nets and $(t, s)$-sequences.
  In: Monte Carlo and Quasi-Monte Carlo Methods in Scientific Computing.
  Lecture Notes in Statististics. Vol. 106. Springer, New York, pp. 299--317.

\bibitem[{Owen(1997{\natexlab{a}})}]{Owen1997a}
Owen, A.~B., 1997{\natexlab{a}}. Monte {C}arlo variance of scrambled net
  quadrature. SIAM J. Numer. Anal. 34~(5), 1884--1910.

\bibitem[{Owen(1997{\natexlab{b}})}]{Owen1997b}
Owen, A.~B., 1997{\natexlab{b}}. Scramble net variance for integrals of smooth
  functions. Ann. Statist. 25~(4), 1541--1562.

\bibitem[{Owen(1998)}]{Owen1998}
Owen, A.~B., 1998. Scrambling {S}obol' and {N}iederreiter-{X}ing points. J.
  Complexity 14~(4), 466--489.

\bibitem[{Owen(2014)}]{Owen2014}
Owen, A.~B., 2014. A constraint on extensible quadrature rules.
  arXiv:1404.5363.

\bibitem[{Rosenblatt(1952)}]{Rosenblatt1952}
Rosenblatt, M., 1952. Remarks on a multivariate transformation. Ann. Math.
  Statist. 23~(3), 470--472.

\bibitem[{Rubin(1987)}]{Rubin1987}
Rubin, D.~B., 1987. {A} noniterative sampling/importance resampling alternative
  to the data augmentation algorithm for creating a few imputations when
  fractions of missing information are modest: {T}he {SIR} algorithm. J. Am.
  Statist. Assoc., 543--546.

\bibitem[{Rubin(1988)}]{Rubin:SIR}
Rubin, D.~B., 1988. Using the {SIR} algorithm to simulate posterior
  distributions. In: Bernardo, J.~M., DeGroot, M.~H., Lindley, D.~V., Smith, A.
  F.~M. (Eds.), Bayesian Statistics 3. Oxford University Press.

\bibitem[{Shiryaev(1996)}]{Probability}
Shiryaev, A.~N., 1996. Probability. Springer.

\bibitem[{Yue and Mao(1999)}]{Yue1999a}
Yue, R.-X., Mao, S.-S., 1999. On the variance of quadrature over scrambled nets
  and sequences. Statist. Prob. Letters 44, 267--280.

\end{thebibliography}

\appendix

\section{Proofs}

\subsection{Proof of Theorem \ref{thm:GenOwen}}\label{app:Thm}

We first prove the following  lemma that plays a key role in the proof of Theorem \ref{thm:GenOwen}. 
\begin{lem}\label{lemma:M} Let $b> 1$ (not necessary an integer), $k$ and $t$ be two integers such that $k\geq t\geq 0$ and $v_m\in [0,b-1]$, $m=0,...,k$. Then, 
\begin{multline}\label{eq:tsN}
 \sum_{m=t}^k v_mb^m\sum_{|u|>0}\, \sum_{|\kappa|>m-t-|u|}\sigma_{u,\kappa}^2\leq \sum_{m=t}^{k} v_{m} b^m\sum_{|u|>0}\sum_{|\kappa|>k-t-|u|}\sigma^2_{u,\kappa}\\
 +b^{k}\sum_{|u|>0}\frac{1}{b^{k-t-|u|}} \sum_{|\kappa|\leq k-t-|u|} \sigma_{u,\kappa}^2 b^{|\kappa|}
\end{multline}
where we use the convention that empty sums are null.
\end{lem}
\begin{proof}

For $u\subseteq\mathcal{S}$ and for $l\in\mathbb{Z}$, let $\tilde{\sigma}^2_{u,l}= \sum_{\kappa: |\kappa|=l}\sigma^2_{u,\kappa}$ if $l\geq 0$ and $\tilde{\sigma}^2_{u,l}=0$ otherwise. To simplify the notations, let $k_t=k-t$ and $v_m'=v_{m+t}$.  Then,
\begin{align*}
 \sum_{m=t}^k v_m b^m\sum_{|u|>0}\, \sum_{|\kappa|>m-t-|u|}\sigma_{u,\kappa}^2=b^t \sum_{m=0}^{k_t} v'_{m} b^m\sum_{|u|>0}\, \sum_{l>m-|u|}\tilde{\sigma}_{u,l}^2.
\end{align*}
Let $N_t=\sum_{m=t}^{k} v_{m} b^m$ so that, using \eqref{eq:sigma}, we have
\begin{align}
 \sum_{m=t}^k v_m b^m\sum_{|u|>0}\, \sum_{|\kappa|>m-t-|u|}\sigma_{u,\kappa}^2= N_t\sigma^2-b^t\sum_{m=0}^{k_t} v'_{m}b^m \sum_{|u|>0}\,\sum_{l\leq m-|u|} \tilde{\sigma}_{u,l}^2.\label{eq:int}
\end{align}
 In order to study the second term  of \eqref{eq:int}, let $u\subseteq\mathcal{S}$ be such that $k_t\geq|u|$. Then,
\begin{align*}
\sum_{m=0}^{k_t} v'_{m}b^m \sum_{l\leq m-|u|} \tilde{\sigma}_{u,l}^2&= \sum_{m=|u|}^{k_t} v'_{m}b^m \,\sum_{l=0}^{m-|u|} \tilde{\sigma}_{u,l}^2=\sum_{m=0}^{k_t-|u|} v'_{m+|u|}b^{m+|u|} \,\sum_{l=0}^{m} \tilde{\sigma}_{u,l}^2.
\end{align*}
Since
$$
\sum_{m=0}^{k_t-|u|} v'_{m+|u|}b^{m+|u|}\,\sum_{l=0}^{m}\tilde{\sigma}^2_{u,l}=\sum_{l=0}^{k_t-|u|}\tilde{\sigma}_{u,l}^2\, \sum_{m=l}^{k_t-|u|} v'_{m+|u|}b^{m+|u|},
$$
with 
$
b^t\sum_{m=l}^{k_t-|u|} v'_{m+|u|}b^{m+|u|}= N_t-\sum_{m=0}^{l+t+|u|-1}v_mb^m$, we obtain
\begin{align*}
b^t\sum_{m=0}^{k_t} v'_m b^m \sum_{l=0}^{m-|u|}\tilde{\sigma}_{u,l}^2&=\sum_{l=0}^{k_t-|u|}\Big(N_t-\sum_{m=0}^{l+t+|u|-1}v_mb^m\Big)
\tilde{\sigma}_{u,l}^2.
\end{align*} 
Therefore, using \eqref{eq:int} and the convention that empty sums are null,  
\begin{align*}
 \sum_{m=t}^k v_mb^m\sum_{|u|>0}\, \sum_{|\kappa|>m-t-|u|}\sigma_{u,\kappa}^2&= N_t\sigma^2 -\sum_{|u|>0}\,\sum_{l=0}^{k_t-|u|}\Big(N_t-\sum_{m=0}^{l+t+|u|-1}v_mb^m\Big)\tilde{\sigma}_{u,l}^2\\
&=  N_t\sum_{|u|>0}\,\sum_{l>k_t-|u|}\tilde{\sigma}^2_{u,l}+\sum_{|u|>0}\sum_{l=0}^{k_t-|u|}\tilde{\sigma}_{u,l}^2\sum_{m=0}^{l+t+|u|-1}v_mb^m.
\end{align*}

Finally, since $v_m\leq b-1$, we have, for $u\subseteq\mathcal{S}$ such that $k_t\geq|u|$,
\begin{align*}
\sum_{l=0}^{k_t-|u|}\tilde{\sigma}_{u,l}^2\sum_{m=0}^{l+|u|+t-1}v_mb^m&\leq (b-1)\sum_{l=0}^{k_t-|u|}\tilde{\sigma}_{u,l}^2\frac{b^{l+|u|+t}-1}{b-1}\leq b^{k}\Big(\frac{1}{b^{k-t-|u|}} \sum_{l=0}^{k_t-|u|}\tilde{\sigma}_{u,l}^2 b^{l}\Big).
\end{align*}
This shows that
\begin{multline*}
 \sum_{m=t}^k v_mb^m\sum_{|u|>0}\, \sum_{|\kappa|>m-t-|u|}\sigma_{u,\kappa}^2\leq  \sum_{m=t}^{k} v_{m} b^m\sum_{|u|>0}\sum_{|\kappa|>k-t-|u|}\sigma^2_{u,\kappa}\\
 +b^{k}\sum_{|u|>0}\frac{1}{b^{k-t-|u|}} \sum_{l\leq k-t-|u|} \tilde{\sigma}_{u,l}^2 b^{l}
\end{multline*}
and the proof of the lemma is complete.

\end{proof}

To prove  Theorem \ref{thm:GenOwen}, and following the proof of \citet[][Lemma 4.11, p.56]{Niederreiter1992},  we decompose $\{\tilde{\bx}^n\}_{n=0}^{N-1}$, $N\geq 1$, into  scrambled $(\lambda_m,t,m,s)$-nets $\tilde{P}_m$, $m=t,\dots,k$, and a remaining set $\tilde{P}$ that contains strictly less than $b^t$ points. We recall that   $k$ is the largest power of $b$ such that $b^k\leq N$.
 
To construct this partition of $\{\tilde{\bx}^n\}_{n=0}^{N-1}$, let $N=\sum_{m=0}^ka_m b^m$ be the expansion of $N$ in base $b\geq 2$, with $a_m\in\{0,...,b-1\}$ and $a_k\neq 0$. Then, let $\tilde{P}_k=\{\tilde{\bx}^n\}_{n=0}^{a_kb^k-1}$ and, for $0\leq m\leq k-1$,  let $\tilde{P}_m$  be the point set made of the $\tilde{\bx}^n$'s with $\sum_{h=m+1}^k a_h b^h\leq n<\sum_{h=m}^k a_h b^h$. By definition of a $(t,s)$-sequence, $\tilde{P}_m$  is a scrambled $(a_m,t,m,s)$-nets in base $b\geq 2$ for $m=t,\dots,k$ while $\tilde{P}=\cup_{m=0}^{t-1}\tilde{P}_m$ has cardinality strictly smaller than $b^t$.

Using this decomposition of $\{\tilde{\bx}^n\}_{n=0}^{N-1}$ we have, using the convention that empty sums are equal to zero, 
\begin{align}
\var\bigg(\frac{1}{N}\sum_{n=0}^{N-1}\varphi(\tilde{\bx}^{n})\bigg)& =\var\bigg(\frac{1}{N} \sum_{\tilde{\bx}\in \tilde{P}}\varphi(\tilde{\bx})+\frac{1}{N}\sum_{m=t}^{k} \sum_{\tilde{\bx}\in \tilde{P}_{m}}\varphi(\tilde{\bx})\bigg)\notag\\
&\leq \Bigg(\frac{1}{N}\bigg\{\var\Big(\sum_{\tilde{\bx} \in \tilde{P}}\varphi(\tilde{\bx})\Big)\bigg\}^{1/2}+\frac{1}{N}\sum_{m=t}^{k}\bigg\{\var\Big(\sum_{\tilde{\bx} \in \tilde{P}_{m}}\varphi(\tilde{\bx})\Big)\bigg\}^{1/2}\Bigg)^{2}\notag\\
&=\frac{1}{N^2}\var\bigg(\sum_{\tilde{\bx} \in \tilde{P}}\varphi(\tilde{\bx})\bigg)
+\bigg(\frac{1}{N}\sum_{m=t}^{k}\bigg\{\var\Big(\sum_{\tilde{\bx} \in \tilde{P}_{m}}\varphi(\tilde{\bx})\Big)\bigg\}^{1/2}\bigg)^2\notag\\
&+\frac{2}{N}\bigg\{\var\Big(\sum_{\tilde{\bx} \in \tilde{P}}\varphi(\tilde{\bx})\Big)\bigg\}^{1/2}
\frac{1}{N}\sum_{m=t}^{k}\bigg\{\var\Big(\sum_{\tilde{\bx} \in \tilde{P}_{m}}\varphi(\tilde{\bx})\Big)\bigg\}^{1/2}\label{eq:bt}.
\end{align}
To bound  the first term of \eqref{eq:bt}, let  $\tilde{J}\subset\{0,\dots,N-1\}$ be such that $n\in \tilde{J}$ if and only if $\tilde{\bx}^n\in \tilde{P}$. Then, note that
\begin{align*}
\var\bigg(\sum_{\tilde{\bx} \in \tilde{P}}\varphi(\tilde{\bx})\bigg)\leq \bigg(\sum_{n \in \tilde{J}}\Big\{\var\big(\varphi(\tilde{\bx}^n)\big)\Big\}^{1/2}\bigg)^2=|\tilde{P}|^2\sigma^2<b^{2t}\sigma^2
\end{align*}
and therefore
\begin{equation}\label{eq:adt}
\frac{1}{N^2}\var\bigg(\sum_{\tilde{\bx} \in \tilde{P}}\varphi(\tilde{\bx})\bigg) \leq \frac{b^{2t}}{N^2}\sigma^2.
\end{equation}

To  bound the second term of \eqref{eq:bt} define, for $m\in\{t,\dots,k\}$, $\hat{I}_{m}=(a_m b^{m})^{-1}\sum_{\tilde{\bx}\in \tilde{P}_{m}}\varphi(\tilde{\bx})$ if $m$ is such that $a_m\neq 0$ and set $\hat{I}_m=0$ otherwise. Then,
\begin{align}
&\bigg(\frac{1}{N}\sum_{m=t}^{k}\bigg\{\var\Big(\sum_{\tilde{\bx} \in \tilde{P}_{m}}\varphi(\tilde{\bx})\Big)\bigg\}^{1/2}\bigg)^2= \bigg(\frac{1}{N}\sum_{m=t}^k \bigg\{\var\Big(a_m b^m \hat{I}_m\Big)\bigg\}^{1/2}\bigg)^2\notag\\
&=\frac{1}{N^2}\sum_{m=t}^k \var(a_m b^m\hat{I}_m)+\frac{2}{N^2}\sum_{k\geq m>n\geq t}\var(a_mb^m\hat{I}_m)^{1/2}
\var(a_nb^n\hat{I}_n)^{1/2}.\label{eq:MainB}
\end{align}

Using \eqref{res}  we have, for $m$ such that $a_m\neq 0$,
$$
\var(\hat{I}_{m})\leq \Gamma^{(b)}_{t,s} (a^mb^m)^{-1}\sum_{|u|>0}\, \sum_{|\kappa|>m-t-|u|}\sigma_{u,\kappa}^2
$$ 
and therefore, using Lemma \ref{lemma:M} and the fact that $b^k\leq N$,
$$
\sum_{m=t}^k \var(a_m b^m\hat{I}_m)\leq \Gamma^{(b)}_{t,s} \sum_{m=t}^k a_m b^m\sum_{|u|>0}\, \sum_{|\kappa|>m-t-|u|}\sigma_{u,\kappa}^2 \leq N\,\Gamma^{(b)}_{t,s}\, B^{(k)}_t
$$
where $B^{(k)}_t$ is as in the statement of the theorem. Hence, $
N^{-1}\sum_{m=t}^k\var(a_mb^m\hat{I}_m)\leq \Gamma^{(b)}_{t,s}B_t^{(k)}$.

To study the second term of \eqref{eq:MainB}, let $m>n\geq t$. Then, easy computations show that 
\begin{align*}
\frac{(a_na_mb^{n+m})^{1/2}}{\Gamma^{(b)}_{t,s}}\var(\hat{I}_m)^{1/2}\var(\hat{I}_n)^{1/2}&\leq   \Big(\sum_{|u|>0}\, \sum_{|\kappa|>m-t-|u|}\sigma_{u,\kappa}^2\Big)^{1/2}\Big(\sum_{|u|>0}\, \sum_{|\kappa|>n-t-|u|}\sigma_{u,\kappa}^2\Big)^{1/2}\notag\\
&\leq \frac{1}{2}\Big(\sum_{|u|>0}\, \sum_{|\kappa|>m-t-|u|}\sigma_{u,\kappa}^2+\sum_{|u|>0}\, \sum_{|\kappa|>n-t-|u|}\sigma_{u,\kappa}^2\Big)
\end{align*}
 and therefore
\begin{align*}
&\frac{2}{\Gamma^{(b)}_{t,s}}\sum_{m=t+1}^k \sum_{n=t}^{m-1}\var(a_mb^m \hat{I}_m)^{1/2}\var(a_nb^n \hat{I}_n)^{1/2} \\
&=\frac{2}{\Gamma^{(b)}_{t,s}} \sum_{m=t+1}^k \sum_{n=t}^{m-1}(a_ma_m)^{1/2}b^{\frac{n+m}{2}}\left\{(a_ma_nb^{n+m})^{1/2}
\big\{\var(\hat{I}_m)\big\}^{1/2}\big\{\var(\hat{I}_n))\big\}^{1/2}\right\}\\
&\leq \sum_{m=t+1}^k \sum_{n=t}^{m-1}\Big\{(a_ma_n)^{1/2}b^{\frac{n+m}{2}} \sum_{|u|>0}\, \sum_{|\kappa|>m-t-|u|}\sigma_{u,\kappa}^2+(a_na_m)^{1/2}b^{\frac{m+n}{2}}\sum_{|u|>0}\, \sum_{|\kappa|>n-t-|u|}\sigma_{u,\kappa}^2\Big\}.
\end{align*}
Consequently, since $\sum_{n=t}^{m-1}a_n^{1/2}b^{n/2}\leq c_b b^{m/2}$, with $c_b=\sqrt{b-1}/(b^{1/2}-1)$, we have
\begin{align}\label{eq:Bt0}
\frac{2}{\Gamma^{(b)}_{t,s}}\sum_{m=t+1}^k \sum_{n=t}^{m-1}\var(a_mb^m \hat{I}_m)^{1/2}&\var(a_nb^n \hat{I}_n)^{1/2}\leq  c_b\Big(\sum_{m=t+1}^k  a_mb^m\sum_{|u|>0}\, \sum_{|\kappa|>m-t-|u|}\sigma_{u,\kappa}^2\Big)\notag\\
&+\sum_{m=t+1}^k a_m^{1/2} b^{m/2}\Big(\sum_{n=t}^{m-1}a_n^{1/2} b^{n/2}\sum_{|u|>0}\, \sum_{|\kappa|>n-t-|u|}\sigma_{u,\kappa}^2\Big)
\end{align}
where, by Lemma \ref{lemma:M}, the first term in bracket is bounded by $NB^{(k)}_t$. For the  second term in bracket, we have, using Lemma \ref{lemma:M} (where  $k$ is replaced by $m-1$ and $b$  by $b^{1/2}>1$),
\begin{multline*}
\sum_{m=t+1}^k a_m^{1/2} b^{m/2}\sum_{n=t}^{m-1}a_n^{1/2}b^{n/2}\sum_{|u|>0}\, \sum_{|\kappa|>n-t-|u|}\sigma_{u,\kappa}^2\\
\leq \sum_{m=t+1}^k a_m^{1/2}b^{m/2}\Big(\sum_{n=t}^{m-1}a^{1/2}_nb^{n/2}\sum_{|u|>0}\sum_{|\kappa|>m-1-t-|u|}\sigma^2_{u,\kappa}\Big.\\
\Big.+b^{\frac{m-1}{2}}\sum_{|u|>0}b^{-\frac{m-1-t-|u|}{2}} \sum_{l\leq m-1-t-|u|}\tilde{\sigma}_{u,l}^2 b^{l/2}\Big)
\end{multline*}
where we recall that, for $u\subseteq\mathcal{S}$ and for $l\in\mathbb{Z}$,  $\tilde{\sigma}^2_{u,l}= \sum_{\kappa: |\kappa|=l}\sigma^2_{u,\kappa}$ if $l\geq 0$ and $\tilde{\sigma}^2_{u,l}=0$ otherwise. Then, using again the fact that $\sum_{n=t}^{m-1}a_n^{1/2}b^{n/2}\leq c_b b^{m/2}$, the right-hand side of the last expression is bounded by
\begin{multline}\label{eq:brack2}
 c_b\sum_{m=t+1}^{k} a_{m}b^{m}\sum_{|u|>0}\,\sum_{|\kappa|> m-1-t-|u|}\sigma_{u,\kappa}^2+b^{\frac{t}{2}}\sum_{|u|>0}b^{\frac{|u|}{2}}\,\sum_{m=t+1}^{k}a^{1/2}_mb^{m/2}\sum_{l\leq m-1-t-|u|}\,b^{l/2}\tilde{\sigma}_{u,l}^2
\end{multline}
with the first term   bounded by $N B^{(k)}_{t+1}$ using  Lemma \ref{lemma:M}. 

To simplify the notation in what follows, let $k_t=k-t$ and $a'_m=a_{m+t}$ for $m=t,\dots,k$. Then,  in the same spirit as for  the derivation of the upper bound given in equation \eqref{eq:tsN},  the second term of \eqref{eq:brack2} can be rewritten as 
\begin{align*}
b^{\frac{t}{2}}\sum_{|u|>0}b^{\frac{|u|}{2}}\,\sum_{m=t+1}^{k}a^{1/2}_mb^{m/2}\sum_{l=0}^{m-1-t}\,b^{\frac{l-|u|}{2}}\tilde{\sigma}_{u,l-|u|}^2&=b^{t+\frac{1}{2}}\sum_{|u|>0}\,\sum_{m=0}^{k_t-1}(a'_{m+1}b^{m})^{1/2}\sum_{l=0}^{m}\,b^{\frac{l}{2}}\tilde{\sigma}_{u,l-|u|}^2\\
&=b^{t+\frac{1}{2}}\sum_{|u|>0}\,\sum_{l=0}^{k_t-1}b^{\frac{l}{2}}\,\tilde{\sigma}_{u,l-|u|}^2\sum_{m=l}^{k_t-1}(a'_{m+1}b^{m})^{1/2}
\end{align*}
with $b^{t+\frac{1}{2}}\sum_{m=l}^{k_t-1}(a'_{m+1}b^{m})^{1/2}\leq c_bb^{\frac{k_t+1}{2}}$. Therefore,
\begin{align*}
b^{\frac{t}{2}}\sum_{|u|>0}b^{\frac{|u|}{2}}\,\sum_{m=t+1}^{k}a_mb^{m/2}\sum_{l\leq m-1-t-|u|}\,b^{\frac{l}{2}}\tilde{\sigma}_{u,l}^2&\leq c_b b^{\frac{k_t+1}{2}}\sum_{|u|>0}\,\sum_{l=0}^{k_t-1}b^{\frac{l}{2}}\tilde{\sigma}_{u,l-|u|}^2\\
&=c_bb^{\frac{k_t+1}{2}}\sum_{|u|>0}b^{\frac{|u|}{2}}\sum_{l=0}^{k_t-1-|u|}\,b^{\frac{l}{2}}\tilde{\sigma}_{u,l}^2\notag\\
&\leq c_bb^{k}\sum_{|u|>0} \frac{1}{b^{\frac{k_t-1-|u|}{2}}}\sum_{l=0}^{k_t-1-|u|}\,b^{\frac{l}{2}}\tilde{\sigma}_{u,l}^2.
\end{align*}
We conclude the proof using the fact that  $(c_1^{1/2}+c_2^{1/2})^2\leq 2(c_1+c_2)$.

\subsection{Proof of the bounds given in \eqref{eq:B1}- \eqref{eq:B2}}\label{app:B}

To prove the bound given in \eqref{eq:B1} note that, for any $c\in\mathbb{N}$, $B_c^{(k)}\leq \sigma^2$. In addition, from \eqref{eq:Bt0}, the term $2(\Gamma^{(b)}_{t,s})^{-1}\sum_{m>n\geq t}\var(a_mb^m \hat{I}_m)^{1/2}\var(a_nb^n \hat{I}_n)^{1/2}$ is bounded by
\begin{align*}
c_b N\sigma^2+\sum_{m=t+1}^k a_m^{1/2} b^{m/2}\sum_{n=t}^{m-1}a_n^{1/2} b^{n/2}\sum_{|u|>0}\, \sum_{|\kappa|>n-t-|u|}\sigma_{u,\kappa}^2&\leq \sigma^2\Big(c_bN+\sum_{m=t+1}^k a_m^{1/2}b^{m/2}\sum_{n=0}^{m-1}a^{1/2}_nb^{n/2}\Big)\\
 &\leq 2c_b N\sigma^2
\end{align*}
and therefore the bound in \eqref{eq:B1} follows from \eqref{eq:bt}-\eqref{eq:MainB}. 

To prove the bound for $t=0$, first note that, in this case, $\tilde{P}=\emptyset$. In addition,  a $(0,s)$-sequence in base $b$ exists only if $ b\geq s$ \citep[see][Corollary 4.36, p.141]{dick2010digital} and therefore the gain factors $\Gamma_{u,\kappa}$ are bounded by  $\Gamma_{0,s}^{(b)}=e$. Hence, 
$
\sum_{m=0}^k\var\left(a_mb^m\hat{I}_m\right)\leq \Gamma_{0,s}^{(b)} N\sigma^2
$
and, using \eqref{eq:MainB}, we conclude that
\begin{align*}
\var\bigg(\frac{1}{N}\sum_{n=0}^{N-1}\varphi(\tilde{\bx}^n)\bigg)&\leq \frac{\sigma^2}{N}e\,\left(1+2c_b\right)\leq \frac{\sigma^2}{N}e\,(3+2\sqrt{2}).
\end{align*}

\end{document}